\documentclass[11pt]{article}
\usepackage{fullpage}
\usepackage{subcaption}
\usepackage[usenames,dvipsnames]{xcolor}
\usepackage[colorlinks,citecolor=blue,linkcolor=BrickRed]{hyperref}
\usepackage{latexsym,graphicx,epsfig,color}
\usepackage{amsfonts,amssymb,amsmath,amsthm,amstext}
\usepackage{bbm}
\usepackage{libertine}
\usepackage[libertine]{newtxmath}
\usepackage{enumitem}
\setitemize{itemsep=2pt,topsep=0pt,parsep=0pt}
\usepackage{url,setspace}
\usepackage{multirow}
\usepackage{rotating}
\usepackage{makeidx}
\usepackage{tikz}
\usepackage{accents}
\usepackage{xspace}
\usepackage{algorithm,algorithmicx}
\usepackage[noend]{algpseudocode}
\usepackage{bm}
\usepackage{bbm}
\usepackage{cleveref}
\usepackage[numbers]{natbib}
\usepackage{thmtools,thm-restate}
\usepackage{authblk}
\usepackage{pifont}% http://ctan.org/pkg/pifont
\allowdisplaybreaks

\newcommand{\IGNORE}[1]{}

\usetikzlibrary{decorations.markings}
\usetikzlibrary{arrows}
\tikzstyle{block}=[draw opacity=0.7,line width=1.4cm]
\tikzstyle{graphnode}=[circle, draw, fill=black!20, inner sep=0pt, minimum width=6pt]
\tikzstyle{point}=[circle, draw, fill=black!30, inner sep=0pt, minimum width=1pt]
\tikzstyle{input}=[rectangle, draw, fill=black!75,inner sep=3pt, inner ysep=3pt, minimum width=4pt]
\tikzstyle{unmatched}=[graphnode,fill=black!0]
\tikzstyle{shaded}=[graphnode,fill=black!20]
\tikzstyle{matched}=[graphnode,fill=black!100]  	
\tikzstyle{matching} = [ultra thick]
\tikzset{
	>=stealth',
	pil/.style={
		->,
		thick,
		shorten <=2pt,
		shorten >=2pt,}
}
\tikzset{->-/.style={decoration={
			markings,
			mark=at position .5 with {\arrow{>}}},postaction={decorate}}}
\newcommand{\todo}[1]{{\color{red}\textsl{\small[#1]}\marginpar{\large\textbf{To Do!}}}}

\makeindex

\newtheorem{theorem}{Theorem}[section]

\newtheorem{lemma}[theorem]{Lemma}
\newtheorem{corollary}[theorem]{Corollary}

\theoremstyle{definition}

\newtheorem{defn}[theorem]{Definition}

\newtheorem{rem}[theorem]{Remark}

\newcommand{\E}{\mathbb{E}}

\newcommand{\OPT}{\textsc{OPT}}

\newcommand{\calD}{\mathcal{D}}

\def\eps {\epsilon}

\def \reals {\mathbb{R}}

\newif\ifFULL
\newcounter{note}[section]
\newcommand{\snote}[1]{\refstepcounter{note}$\ll${\bf Sahil~\thenote:}
  {\sf \color{red}  #1}$\gg$\marginpar{\tiny\bf SS~\thenote}}

\newcommand{\rnote}[1]{\refstepcounter{note}$\ll${\bf Ravi~\thenote:}
  {\sf \color{orange}  #1}$\gg$\marginpar{\tiny\bf RR~\thenote}}

\newcommand{\p}{\ensuremath{\mathbf{p}\xspace}}
\newcommand{\emax}{\ensuremath{\mathsf{MinEMax}}\xspace}
\newcommand{\hybrid}{\ensuremath{\mathsf{Hybrid}}\xspace}
\newcommand{\univ}{U\xspace}
\newcommand{\cost}{\ensuremath{\mathsf{cost}}\xspace}

\newcommand{\bec}[1]{&{\text{$\left(\text{by #1} \right)$}}}%\bec
\newcommand{\val}{\text{val}}
\newcommand{\ct}{\cost_{\textsf{Trunc}}}
\newcommand{\Pe}{P_{\textsf{EMax}}}
\newcommand{\Pt}{P_{\textsf{Trunc}}}
\newcommand{\Ph}{P_{\textsf{Hyb}}}
\newcommand{\ce}{\cost_{\textsf{EMax}}}
\newcommand{\ch}{\cost_{\textsf{Hyb}}}
\newcommand{\cro}{\cost_{\textsf{Rob}}}
\newcommand{\cst}{\cost_{\textsf{Stoch}}}
\newcommand{\met}{\ensuremath{\mathcal{D}}}
\newcommand{\trunc}{\ensuremath{\mathsf{TruncatedTwoStage}}\xspace}
\newcommand{\med}{\ensuremath{\mathcal{M}\mathcal{E}\mathcal{D}}\xspace}
\def\drawScen[#1](#2,#3)[#4,#5] {
\draw (#2,#3) circle [radius=1]; \node [below] at (#2, #3) {#1};
\draw [fill=red!30] (#2,#3) rectangle (#2-.75,.4+#3+#4*.35); \node [below] at (#2-.375,#3 + #4*.35 + .4) {\small #4};
\draw [fill=blue!30] (#2,#3) rectangle (#2+.75,#3 +#5*3+.35); \node [below] at (#2+.375,#3 +#5*3 + .35) {\small #5};
}

\DeclareMathOperator*{\argmin}{arg\,min}
\title{
Prepare for the Expected Worst: \\
Algorithms for Reconfigurable Resources Under Uncertainty
}

\author[1]{D Ellis Hershkowitz\thanks{Supported in part by NSF grants CCF-1527110, CCF-1618280, CCF-1814603, CCF-1910588, NSF CAREER award CCF-1750808 and a Sloan Research Fellowship.}}
\author[1]{R. Ravi\thanks{Supported in part by the U. S. Office of
		Naval Research award N00014-18-1-2099, and the U. S. NSF award CCF-1527032.}}
\affil[1]{Carnegie Mellon University}

\author[2]{Sahil Singla\thanks{Much of this work was done while at Carnegie Mellon University. 
Supported in part by Schmidt Foundation and NSF awards CCF-1319811, CCF-1536002, and CCF-1617790.}}
\affil[2]{Princeton University and Institute for Advanced Study}

\date{ \today}

\begin{document}
\maketitle

\setlength{\abovedisplayskip}{2pt}
\setlength{\belowdisplayskip}{2pt}

\begin{abstract}
	In this paper we study how to optimally balance cheap inflexible resources with more expensive, reconfigurable resources despite uncertainty in the input problem. Specifically, we introduce the \emax model to study ``build versus rent'' problems. In our model different scenarios appear independently. Before knowing which scenarios appear, we may build rigid resources that cannot be changed for different scenarios. Once we know which scenarios appear, we are allowed to rent reconfigurable but expensive resources to use across scenarios. Although computing the objective in our model might seem to require enumerating exponentially-many possibilities, we show it is well-estimated by a surrogate objective which is representable by a polynomial-size LP. In this surrogate objective we pay for each scenario only to the extent that it exceeds a certain threshold. Using this objective we design algorithms that approximately-optimally balance inflexible and reconfigurable resources for several NP-hard covering problems. For example, we study minimum spanning and Steiner trees, minimum cuts and facility location variants. Up to constants our approximation guarantees match those of previous algorithms for the previously-studied demand-robust and stochastic two-stage models. Lastly, we demonstrate that our problem is sufficiently general to smoothly interpolate between previous demand-robust and stochastic two-stage problems.
\end{abstract}
\thispagestyle{empty}
\newpage

\setcounter{page}{1}

% !TeX root = arXiv.tex
% !TEX root = arXiv.tex

\section{Introduction}

\IGNORE{
\textbf{Story 1 (Ski rental)}
Consider the algorithmic challenges faced by a skier. To go skiing on a particular day a skier must have skis. However, the skier only has a rough sense of the days on which they will go skiing. At the beginning of the skiing season a skier can choose to pre-rent skis for some subset of the days; however, skis are in high demand and so the skier must reserve these skis now since, come the day of skiing, all rental skis will already be reserved. If a day comes when the skier choose to ski but has not rented skis, a skier can simply buy a personal pair of skis which can be reused to go skiing any given day. Thus, the skier must balance preemptively purchasing one-time-use resources (renting skis) with a later purchase of a reusable resource (buying skis). Overall, then, the total price that the skier pays is the price of pre-renting skis plus the cost of buying skis if there is some day where they want to ski but did not reserve skis. Stated otherwise, a skier pays for renting skis plus the \emph{expected maximum} additional cost of skiing on any given day where this expectation is taken over the probability that the skier skis. Balancing such one-off commodities with more expensive but flexible commodities extends to various problems in network design and facility location.
}

Optimizing for reconfigurable resources under uncertainty formalizes the challenges of balancing expensive, flexible resources with cheap, inflexible ones. For example, such optimization problems formalize the challenges in ``build versus rent'' problems. Concretely, consider the algorithmic challenges faced by an Internet service provider (ISP). An ISP must provide content to its customers while balancing between rigid and reconfigurable resources. In particular, it can build out its own network---a rigid resource---or choose to support traffic on a competitor's network---a flexible resource---at a marked up premium. This latter resource is reconfigurable since an ISP can change which edges in a competitor's network it uses at any given time. To minimize the additional load on its network, the competitor charges the ISP for the maximum extra bandwidth it must support at any given moment.
%If an ISP knew where all demands were going to occur it could simply forgo the premiums of renting a competitor's network and build out its own network to satisfy demands.
Furthermore, an ISP only has probabilistic knowledge of where customer demands will occur: Based on where previous demands have occurred an ISP estimates future demands, but it does not exactly know the future demands. If a demand occurs which the ISP's network cannot service, it
% does not have time to build out its network to immediately support this demand. Rather, it
must use the competitor's network to support it.
%Thus, uncertainty in demands force an ISP to balance spending on a rigid resource---its own network---with a flexible, but expensive, resource---its competitor's network.
Thus, an ISP  balances between rigid and flexible resources in the face of uncertainty, and pays for the cost of its own network plus the cost of supporting the \emph{expected maximum} traffic routed on its competitor's network.

In this paper, we introduce the \emax model to study the algorithmic challenges associated with optimizing reconfigurable resources under uncertainty. In our model we are given a set of \emph{scenarios} that might occur. In the preceding example these scenarios were the sets of possible demands.
%We are allowed to purchase resources both in a non-reconfigurable way without knowing which scenarios occur, and in a reconfigurable way using contracted resources at an increased cost after learning which scenarios occur.
We think of problems in our model as being divided between a first stage where we  ``build'' rigid resources and a second stage where we ``rent'' flexible resources. In particular, in the first stage we can build non-reconfigurable resources without knowing which scenarios occur. In the second stage, each scenario independently realizes according to its specified Bernoulli probability, and we can rent reconfigurable resources at an increased cost to use among any of our scenarios.% The reconfigurable resources can be shared among any of the materialized scenarios.
~For instance, in the preceding example the ISP first built its own network and then, once it learned where demands occurred, it could rent bandwidth to support different demands over time. In fact, this example is exactly our \emax Steiner tree problem.  Thus, the objective we  minimize is the first stage cost plus the \emph{expected maximum} cost of additional reconfigurable resources required for any realized scenario; hence the name of our model.

Since every scenario is an independent Bernoulli, there are exponentially-many ways in which scenarios realize. It is not even clear, then, how to compute the expected second-stage cost. Nonetheless, we provide techniques to simplify and reason about the \emax cost and, therefore, solve various \emax problems.

The primary contributions of our work are as follows.
\begin{enumerate}[itemsep=0ex,label=(\roman*),topsep=0pt,parsep=0pt]
	%\begin{enumerate}\itemsep0em %Change this to change how spaced out the items are
	\item We introduce the \emax model for optimization of reconfigurable resources under uncertainty.
	\item We show that, although evaluating the \emax objective function may seem difficult, a \emax problem can be approximately reduced to a ``\trunc'' problem whose objective is representable by an LP.  \label{contrib:EMaxToTrunc}
	\item Armed with \ref{contrib:EMaxToTrunc}, we adapt various rounding techniques to give approximation algorithms for a variety of two-stage \emax problems including spanning and Steiner trees, cuts and facility location problems.
	\item Lastly, we show that the \emax model  captures the commonly studied two-stage models for optimization under uncertainty: the stochastic and demand-robust models. Indeed, we show that it generalizes a ``\hybrid'' problem that smoothly interpolates between the stochastic and demand-robust models.
\end{enumerate}

\subsection{Related Work}
Significant prior work has been done in two-stage optimization under uncertainty. The two most commonly studied models are the stochastic model~\cite{RS-IPCO04,GPRS-STOC04,SS2006} and the demand-robust model~\cite{DGRS-FOCS05,AGGN-SODA08,GNR-ICALP10,Golovin2015}. In the \textbf{stochastic two-stage model} a probability distribution is given over scenarios, and our objective is to minimize the \emph{expected} total cost. In the {\bf demand-robust} two-stage model we always pay for the \emph{worst-case} scenario given our first stage solution.

\emph {Distributionally robust optimization} (DRO)~\cite{Scarf1958,GohSim2010,DelageYe2010,BBC2011} captures a problem similar to our own. In the DRO model we are given a distribution along with a ball of ``nearby'' distributions and must pay the \emph{worst-case expectation} over all distributions. Similarly to our own model, DRO generalizes both the stochastic and demand-robust two-stage models. Our model can be seen as a ``flip'' of the DRO model: while the DRO model takes the worst-case over distributions our model takes a distribution over worst cases. Like DRO, our model is also sufficiently general to capture stochastic and demand-robust optimization. An exciting forthcoming result \cite{LS-STOC19}---which shows that approximation algorithms are possible in DRO---complements our approximation algorithms in \emax.

A well studied measure for risk-aversion from stochastic programming is \emph{conditional value at risk (CVaR)} \cite{acerbi2002expected}. Roughly, CVaR gives the average cost in the worst-case case $\alpha$ tail of a distribution.  A notable recent work in CVaR presents a data-driven approach to two-stage risk aversion \cite{jiang2018risk}. %In this model, historical data is used to construct a confidence set of possible distributions and then one optimizes against the worst-case distribution in this confidence set. 
Theorem 1 of this work is reminiscent of our reduction of \emax to \hybrid; this theorem shows that their objective function can be reformulated as a combination of the CVaR cost and the worst-case distribution. We emphasize that while CVaR might appear similar to the \trunc metric studied in this work, these two metrics are distinct and not readily comparable. Two salient differences are (1) the threshold in the \trunc objective is the minimizing threshold while in CVaR the threshold is fixed and (2) the \trunc objective sums up the truncated cost over a set of Bernoulli random variables whereas CVaR takes a truncated average cost with respect to a single distribution. We also note that, to our knowledge, CVaR has not been studied in an approximation algorithms context, the focus of this work.

Several additional models for optimization under uncertainty---some of which even interpolate between stochastic and demand-robust---have also been studied. A series of papers~\cite{S-SODA07,S-OR09,S-SODA11} has also examined various models of two-stage optimization which capture risk-aversion. Notably, the model of~\cite{S-SODA11} interpolates between stochastic and demand-robust while also accommodating \emph{black-box} distributions. Other papers~\cite{GPRS-STOC04} have also studied algorithms for stochastic optimization given access to black-box distributions. There has also been work on two-stage stochastic models in which---as in our model---independent stochastic outcomes factor prominently. For example, Immorlica et al.~\cite{IKMM-SODA04} study a two-stage stochastic model in which ``clients'' each activate independently and the realized scenario consists of all activated clients. The primary difference between this model and our own is that in our model entire scenarios---rather than clients---activate independently. Moreover, reconfigurability of resources does not factor whatsoever into this model.

%Our problem also bears some similarities to {\em centrum} problems~\cite{Tamir01,CS-arXiv17} whose objectives interpolate between a sum and a max. Similarly, previous work in scheduling~\cite{BP03,KMPS05} in which the $L_p$ norm of the vector of completion times for $1 < p < \infty$ has been proposed to trade-off between the sum of completion times objective $p=1$ and the makespan objective $p = \infty$ is reminiscent of our own.

Lastly, a series of work \cite{lai1976maximally,meilijson1979convex,bertsimas2004probabilistic,dhara2017polynomial} have made use of the bound we use in our reduction of \emax to \trunc in settings distinct from our own. For example, \cite{bertsimas2004probabilistic} shows that this bound tightly estimates the value of the optimal solution in an optimization problem where the cost function is random and only marginal distributions for the coefficients of the cost function are known. Unlike our work, these works are not concerned with approximation algorithms for two-stage NP-hard problems.

\subsection{Models}\label{sec:models}
We now formally define  our new \emax model and the prior models that we generalize. We study two-stage covering problems, defined as follows.

\paragraph{\textbf{Two-Stage Covering.}}
Let $\univ$ be the universe of \emph{clients} (or demand requirements), and let $X$ be the set of \emph{elements} that we can purchase. Every scenario $S_1,S_2,\ldots,S_m$ is a subset of clients. Let $sol(S_s)$ for $s \in [m]$ denote the sets in $2^X$ which are feasible to \emph{cover} scenario $S_s$. In covering problems if $A \subseteq B$ and $A \in sol(S_s)$, then $B \in sol(S_s)$. We are also given a cost function $\cost: 2^X \times 2^X \rightarrow \reals$. For a given a specification of $\cost$,  scenarios, clients, and feasibility constraints, we must find a set of elements $X_1\subseteq X$ to be bought in the first stage, and a set of elements $X_2^{(s)}\subseteq X$ to be bought in the second stage s.t.  $X_1 \cup  X_2^{(s)} \in sol(S_s)$ for every $s$. Our goal is to find a solution of minimal cost where the cost of a solution is discussed below.

This paper makes the common assumption that $\cost$ is \emph{linear}, i.e.,  $\cost(X_1, X_2^{(s)})$ equals $\cost(\emptyset, X_2^{(s)}) + \cost(X_1,  \emptyset)$ for any $X_1, X_2^{(s)} \subseteq 2^X$.
Let $\pmb{X_2} := (X_2^{(1)}, \ldots, X_2^{(m)})$; throughout the paper a bold variable  denotes a vector.

%As a concrete example, in the star-graph covering problem $X$ is the set of edges, and each scenario $S_{s}$ is associated with a  node $v_s$. A subset of edges $X_1 \cup  X_2^{(s)}$ is feasible for $S_{s}$ if it connects $v_s$ to the root $r$. Here $\cost(X_1,X_2^{(s)}) := \sum_{e\in X_1}c_e + \sum_{e\in X_2^{(s)}} \sigma \cdot c_e$, where $\sigma > 1$ is a known second stage cost inflation factor.

We now describe and discuss how different cost functions yield different two-stage covering models.

\paragraph{\textbf{Prior Models.}}
%\paragraph{Demand Robust Model.}
In the \emph{demand-robust} two-stage covering model the cost of solution $(X_1, \pmb{X_2})$ is the maximum cost  over all  the scenarios:
\begin{align} \label{eq:RobModel}
\cro(X_1, \pmb{X_2}) := \max_{s \in [m]}\Big\{\cost(X_1,X_2^{(s)}) \Big\}.
\end{align}
%Notice that in this model we always pessimistically assume that we pay for the most expensive scenario.

%Note that we pay for all the elements in $X_0$ even though some of them may not be required in the solution for any one fixed scenario.
%\paragraph{Stochastic Model.}
In the \emph{stochastic} two-stage covering model we are given a probability distribution $\calD$ over $m$ scenarios with which exactly one of them realizes; i.e.\ $\sum_{s \in [m]} \mathcal{D}(s) = 1$. The cost of solution $(X_1, \pmb{X_2})$ is the expected cost:
\begin{align} \label{eq:StochModel}
\cst(X_1, \pmb{X_2}) := \E_{s\sim \calD}[\cost(X_1,X_2^{(s)})].
\end{align}

\paragraph{\textbf{Our New \emax Model.}} In the $\emax$ two-stage covering model we are given probabilities  $\p = \{p_1, \ldots, p_m\}$ with which each scenario {\em independently} realizes. The cost of solution $(X_1, \pmb{X_2})$ is the expected maximum cost among the realized scenarios:
\begin{align} \label{eq:emaxModel}
\ce(X_1, \pmb{X_2}) := \E_{A \sim \p} \Big[\max_{s \in A} \Big\{\cost(X_1,X_2^{(s)}) \Big\}\Big]
\end{align}
where $A$ contains each $s$ independently w.p. $p_s$. To avoid confusion, we reiterate that  unlike the stochastic model, in \emax multiple scenarios may simultaneously appear in $A$ because each of them independently realizes. We shall assume without loss of generality that $\sum_s p_s \geq 1$ throughout this paper since one can always ensure this without affecting solutions to the problem by adding dummy scenarios of cost $0$ and probability $1$.

As a concrete example these models, consider the following star covering problem. We are given a star graph with root $r$ and leaves $v_1, \ldots, v_m$. Each edge $e_i = (r, v_i)$ can be purchased in the first stage at cost $c_i$ and in the second stage at an inflated  cost  $\sigma \cdot c_s$ for $\sigma > 1$. Our goal is to connect $r$ to an  \emph{unknown} vertex $v_s$ with minimum total two-stage cost. In particular, $v_s$ is only revealed after we purchase our first-stage edges, $X_1$, at which point we must purchase $e_s$ in a second stage at cost $\sigma \cdot c_s$ if $e_s$ was not already purchased in the first stage. In all three models we initially buy some set of edges. In the stochastic version of this problem a single $v_s$ then appears according to a distribution and we must pay to connect $v_s$ if we have not already. In the demand-robust version of this problem, $v_s$ is always chosen so as to maximize our second stage cost. However, in our \emax version of this problem several $v_s$ appear and we must pay for a budget of reconfigurable edge resource to be reused for every $v_s$. See \Cref{fig:star} for an illustration.

\tikzset{nodeStyle/.style={circle,draw=black, minimum size=17,inner sep=2, fill=white}}
\tikzset{edgeStyle/.style={draw=black, above, text=red!70}}

\newcommand{\prMinSize}{10.5}
\newcommand{\opacVal}{.1}
\newcommand{\nodeStyleDefault}{circle} %{circle,draw=blue,fill=blue!20!, minimum size=17,inner sep=0}
\newcommand{\edgeStyleDefault}{edgeStyle/.style={above}}

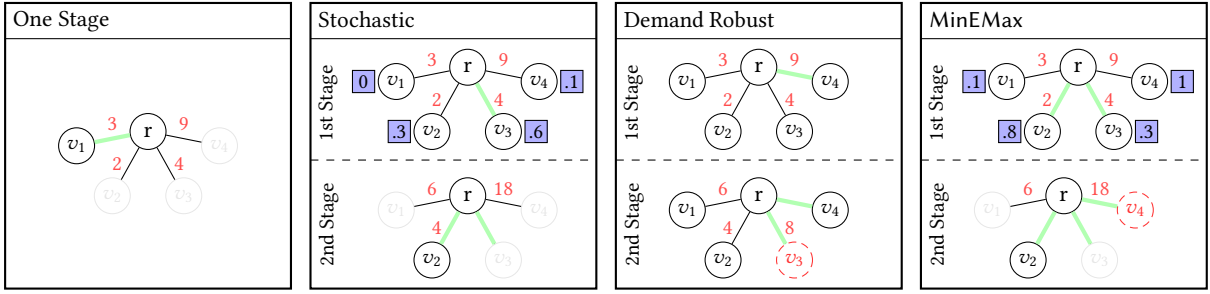
\begin{figure}
	\centering
	\begin{subfigure}[t]{.24\linewidth}
		\centering
		%ONE STAGE
		\begin{tikzpicture}[scale=.19,every node/.style={scale=0.8}]
		\draw [thick] (0,0) rectangle (20,20); \node [below right] at (0,20) {One Stage};
		
		\draw (0,17.5)--(20,17.5);
		
		\node[nodeStyle] (r) at (10,11) {r};
		\node[nodeStyle] (S1) at (5,10) {\small $v_1$}; \draw[edgeStyle]  (r) edge node{\small 3} (S1); \draw[edgeStyle, green!30, ultra thick]  (r) edge (S1);
		\node[nodeStyle, opacity=\opacVal] (S2) at (7.5,6.5) {\small $v_2$};  \draw[edgeStyle, left]  (r) edge node{\small 2} (S2);
		\node[nodeStyle, opacity=\opacVal] (S3) at (12.5,6.5) {\small $v_3$};  \draw[edgeStyle, right]  (r) edge node{\small 4} (S3);
		\node[nodeStyle, opacity=\opacVal] (S4) at (15,10) {\small $v_4$};  \draw[edgeStyle]  (r) edge node{\small 9} (S4);
		\end{tikzpicture}
		\caption{If the scenario to be covered is known to be $v_1$, the problem is trivial.}\label{sfig:oneStage}
	\end{subfigure}
%	STOCHASTIC
	\begin{subfigure}[t]{.24\linewidth}
		\centering
		\begin{tikzpicture}[scale=.19,every node/.style={scale=0.8}]
		\draw [thick] (0,0) rectangle (20,20); \node [below right] at (0,20) {Stochastic};
		
		\draw (0,17.5)--(20,17.5);
		
		\node[nodeStyle] (r) at (11,15.5) {r};
		\node[nodeStyle] (S1) at (6,14.5) {\small $v_1$}; \draw[edgeStyle]  (r) edge node{\small 3} (S1); \node [draw=black, fill=blue!30, left=10, inner sep = 0, minimum size =\prMinSize] at (S1) {\small 0};
		\node[nodeStyle] (S2) at (8.5,11) {\small $v_2$};  \draw[edgeStyle, left]  (r) edge node{\small 2} (S2); \node [draw=black, left=10,fill=blue!30, inner sep = 0, minimum size =\prMinSize] at (S2) {\small .3};
		\node[nodeStyle] (S3) at (13.5,11) {\small $v_3$};  \draw[edgeStyle, right]  (r) edge node{\small 4} (S3); \node [draw=black, right=10,fill=blue!30, inner sep = 0, minimum size =\prMinSize] at (S3) {\small .6}; \draw[edgeStyle, green!30, ultra thick]  (r) edge (S3);
		\node[nodeStyle] (S4) at (16,14.5) {\small $v_4$};  \draw[edgeStyle]  (r) edge node{\small 9} (S4); \node [draw=black, right=10,fill=blue!30, inner sep = 0, minimum size =\prMinSize] at (S4) {\small .1};
		
		\draw [dashed] (0,9)--(20,9);
		
		\node [rotate=90, below] at (0,4.5) {\small 2nd Stage};
		\node[nodeStyle] (r) at (11,6.5) {r};
		\node[nodeStyle, opacity=\opacVal] (S1) at (6,5.5) {\small $v_1$}; \draw[edgeStyle]  (r) edge node{\small 6} (S1);
		\node[nodeStyle] (S2) at (8.5,2) {\small $v_2$};  \draw[edgeStyle, left]  (r) edge node{\small 4} (S2); \draw[edgeStyle, green!30, ultra thick]  (r) edge (S2);
		\node[nodeStyle, opacity=\opacVal] (S3) at (13.5,2) {\small $v_3$};  \draw[edgeStyle, right]  (r) edge (S3); \draw[edgeStyle, green!30, ultra thick]  (r) edge (S3);
		\node[nodeStyle, opacity=\opacVal] (S4) at (16,5.5) {\small $v_4$};  \draw[edgeStyle]  (r) edge node{\small 18} (S4);
		
		\node [rotate=90, below] at (0,13) {\small 1st Stage};
		\end{tikzpicture}
		\caption{Exactly one scenario realizes according to a probability distribution.} \label{sfig:stochastic}
	\end{subfigure}
%	DEMAND ROBUST
	\begin{subfigure}[t]{.24\linewidth}
		\centering
		\begin{tikzpicture}[scale=.19,every node/.style={scale=0.8}]
		\draw [thick] (0,0) rectangle (20,20); \node [below right] at (0,20) {Demand Robust};
		
		\draw (0,17.5)--(20,17.5);
		
		\node[nodeStyle] (r) at (10,15.5) {r};
		\node[nodeStyle] (S1) at (5,14.5) {\small $v_1$}; \draw[edgeStyle]  (r) edge node{\small 3} (S1);
		\node[nodeStyle] (S2) at (7.5,11) {\small $v_2$};  \draw[edgeStyle, left]  (r) edge node{\small 2} (S2);
		\node[nodeStyle] (S3) at (12.5,11) {\small $v_3$};  \draw[edgeStyle, right]  (r) edge node{\small 4} (S3);
		\node[nodeStyle] (S4) at (15,14.5) {\small $v_4$};  \draw[edgeStyle]  (r) edge node{\small 9} (S4); \draw[edgeStyle, green!30, ultra thick]  (r) edge (S4);
		
		\draw [dashed] (0,9)--(20,9);
		
		\node [rotate=90, below] at (0,4.5) {\small 2nd Stage};
		\node[nodeStyle] (r) at (10,6.5) {r};
		\node[nodeStyle] (S1) at (5,5.5) {\small $v_1$}; \draw[edgeStyle]  (r) edge node{\small 6} (S1);
		\node[nodeStyle] (S2) at (7.5,2) {\small $v_2$};  \draw[edgeStyle, left]  (r) edge node{\small 4} (S2);
		\node[nodeStyle, color=red!80, fill=white, densely dashed] (S3) at (12.5,2) {\small $v_3$};  \draw[edgeStyle, right]  (r) edge node{\small 8} (S3); \draw[edgeStyle, green!30, ultra thick]  (r) edge (S3);
		\node[nodeStyle] (S4) at (15,5.5) {\small $v_4$};  \draw[edgeStyle]  (r) edge (S4); \draw[edgeStyle, green!30, ultra thick]  (r) edge (S4);
		
		\node [rotate=90, below] at (0,13) {\small 1st Stage};
		\end{tikzpicture}
		\caption{Given a first stage solution, adversary chooses the costliest scenario.}\label{sfig:demandRob}
	\end{subfigure}
%	EMAX
	\begin{subfigure}[t]{.24\linewidth}
		\centering
		\begin{tikzpicture}[scale=.19,every node/.style={scale=0.8}]
		\draw [thick] (0,0) rectangle (20,20); \node [below right] at (0,20) {\emax};
		
		\draw (0,17.5)--(20,17.5);
		
		\node[nodeStyle] (r) at (11,15.5) {r};
		\node[nodeStyle] (S1) at (6,14.5) {\small $v_1$}; \draw[edgeStyle]  (r) edge node{\small 3} (S1); \node [draw=black, left=10,fill=blue!30, inner sep = 0,, minimum size=\prMinSize] at (S1) {\small .1};
		\node[nodeStyle] (S2) at (8.5,11) {\small $v_2$};  \draw[edgeStyle, left]  (r) edge node{\small 2} (S2); \node [draw=black, left=10,fill=blue!30, inner sep = 0, minimum size=\prMinSize] at (S2) {\small .8}; \draw[edgeStyle, green!30, ultra thick]  (r) edge (S2);
		\node[nodeStyle] (S3) at (13.5,11) {\small $v_3$};  \draw[edgeStyle, right]  (r) edge node{\small 4} (S3); \node [draw=black, right=10,fill=blue!30, inner sep = 0, minimum size=\prMinSize] at (S3) {\small .3}; \draw[edgeStyle, green!30, ultra thick]  (r) edge (S3);
		\node[nodeStyle] (S4) at (16,14.5) {\small $v_4$};  \draw[edgeStyle]  (r) edge node{\small 9} (S4); \node [draw=black, right=10,fill=blue!30, inner sep = 0, minimum size=\prMinSize] at (S4) {\small 1};
		
		\draw [dashed] (0,9)--(20,9);
		
		\node [rotate=90, below] at (0,4.5) {\small 2nd Stage};
		\node[nodeStyle] (r) at (10,6.5) {r};
		\node[nodeStyle, opacity=\opacVal] (S1) at (5,5.5) {\small $v_1$}; \draw[edgeStyle]  (r) edge node{\small 6} (S1);
		\node[nodeStyle] (S2) at (7.5,2) {\small $v_2$};  \draw[edgeStyle, left]  (r) edge (S2);  \draw[edgeStyle, green!30, ultra thick]  (r) edge (S2);
		\node[nodeStyle, opacity=\opacVal] (S3) at (12.5,2) {\small $v_3$};  \draw[edgeStyle, right]  (r) edge (S3);  \draw[edgeStyle, green!30, ultra thick]  (r) edge (S3);
		\node[nodeStyle, color=red!80, fill=white, densely dashed] (S4) at (15,5.5) {\small $v_4$};  \draw[edgeStyle]  (r) edge node{\small 18} (S4);  \draw[edgeStyle, green!30, ultra thick]  (r) edge (S4);
		
		\node [rotate=90, below] at (0,13) {\small 1st Stage};
		\end{tikzpicture}
		\caption{Given a first stage solution, adversary chooses the costliest realized scenario.}\label{sfig:emax}
	\end{subfigure}
	\caption{Star graph \emax for $m=4$. Green edges: edges bought by solution. $e_i$ labeled by its cost in each stage for $\sigma = 2$. Non-opaque second-stage node: realized scenario. Blue square: probability of scenario. Dashed red nodes: nodes chosen by an adversary.}\label{fig:star}
\end{figure}

\subsection{Technical Results and Intuition}
We now  discuss our technical results.
As earlier noted, capturing the \emax objective is challenging: scenarios may realize in exponentially-many ways and so it seems that even computing the objective is computationally infeasible. We solve this issue by showing that to solve a \emax problem, $\Pe$, it suffices to solve its \trunc version, $\Pt$. A \trunc problem is identical to a \emax problem but the cost of a solution $(X_1, \pmb{X_2})$ is its truncated sum:
\begin{align}
\textstyle{ \ct(X_1, \pmb{X_2}) := \min_B \Big[ B + \sum_{s \in [m]} p_s\cdot (\cost(X_1,X_2^{(s)}) - B)^+\Big]. } \label{eq:truncObjective}
\end{align}
We will later see that $\Pt$ can be represented by an LP and, therefore, can be efficiently approximated by various rounding techniques. The following theorem shows that to approximate a \emax problem, it suffices to consider its \trunc version.

\begin{restatable}{theorem}{EmaxToCostB}\label{thm:EmaxToCostB} Let $\Pe$ be a \emax problem and let $\Pt$ be its corresponding \trunc problem. An $\alpha$-approximation algorithm for $\Pt$ is a $\left ( \frac{\alpha}{1-1/e} \right)$-approximation algorithm for $\Pe$.
	%where $\cost_B(X_1,X_2) = 0$ if $\cost(X_1,X_2) \leq B$, and  $\cost_B(X_1,X_2)=\cost(X_1,X_2)$ otherwise.
\end{restatable}
\noindent\textbf{Intuition.} The main observation we use to show this theorem is that a set of expensive scenarios with large total probability mass dominates the cost of a given \emax solution. We illustrate this observation with an example. Let $(X_1, \pmb{X_2})$ be a solution for a \emax problem. Now WLOG let $\cost(X_1, X_2^{(s)} ) \geq \cost(X_1, X_2^{(s+1)})$ for all $s$, i.e., the $s$th scenario is more expensive than the $(s+1)$th scenario for our solution. Let $M:= [k]$ be the indices of the first $k$ scenarios such that $\sum_{s \leq k} p_s$ is large; say, at least $1$. Let the border $B:=\cost(X_1, X_2^{(k)} )$ be the cost of the least expensive scenario with an index in $M$. Because there is a great deal of probability mass among scenarios in $M$ we know that with large probability some scenario in $M$ will always appear. Whenever a scenario of cost less than $B$ appears we know that with good probability something in $M$ has also appeared of greater cost. Thus, as far as the expected max is concerned, a scenario that costs less than $B$ can be ignored. Lastly, while it is not immediately clear how to represent $\ct$ function  in an LP, we show using a simple convexity argument how this can accomplished.
%that $B$ is representable linearly by the above $\ct$ function.

%\subsubsection{Approximations for \emax Covering Problems.}
Next, we design approximation algorithms for two-stage covering problems in the \emax model.

\begin{restatable}{theorem}{Applications}\label{thm:Applications}
	For two-stage covering problems there exist polynomial-time approximation algorithms with the following guarantees. \footnote{%We note that our results for non-\emax problems in \Cref{thm:Applications} follow from the fact that \emax generalizes the stochastic and demand-robust problems. All MST results are w.h.p.\ in expectation; in particular, our algorithm with high probability returns a spanning tree of good expected cost and with the remaining probability returns something which may not be a spanning tree.
		The $O(1)$ in the $k$-center approximation is roughly $57$.}
	\def\arraystretch{1.4}%Change this to pad out the table
	
	\begin{center}
		\begin{tabular}{|c|c|c|c|c|c|}%c|c|c|c|}
			\hline
			\emax Problem  &  UFL & Steiner tree & MST & Min-cut & $k$-center \\
			\hline
			Approximation \ \quad & ${ \frac{8}{1-1/e}}$ & ${\frac{30}{1-1/e}}$ & ${O(\log n + \log m})$  & ${\frac{4}{1-1/e}} $ &  ${O(1)}$ \\
			\hline
		\end{tabular}
	\end{center}
\end{restatable}

\noindent\textbf{Intuition.} Our earlier Theorem \ref{thm:EmaxToCostB} demonstrated that to solve a \emax problem, $\Pe$, we need to only solve its \trunc version, $\Pt$. While it is not clear how to represent $\Pe$ with an LP, $\Pt$ can be represented with an LP. Furthermore, by adapting previous two-stage optimization rounding techniques to the \trunc setting, we are able to approximately solve the \trunc versions of uncapacitated facility location (UFL), Steiner tree, minimum spanning tree (MST), and min-cut.

We use different techniques to give an approximation algorithm for $k$-center. The intuition for our $k$-center proof is similar to that of Theorem \ref{thm:EmaxToCostB}: Truncated costs approximate \emax cost. However, for $k$-center we truncate more aggressively. Rather than truncating costs of scenarios, we truncate distances in the input metric. To do this, we draw on methods of Chakrabarty and Swamy \cite{CS-arXiv17}.\footnote{We also note here that, unlike the previous problems we study, the cost function in $k$-center is not linear as described in \S\ref{sec:models}.}

It is also worth noting that Anthony et al.~\cite{AGGN-SODA08} proved hardness of approximation for a two-stage $k$-center problem. In particular, they show stochastic $k$-center where scenarios consist of \emph{multiple} clients is as hard to approximate as dense $k$-subgraph. Thus, since our \emax model generalizes the stochastic model, we restrict our attention in $k$-center to scenarios consisting of \emph{single} clients; otherwise our problem would be prohibitively hard to approximate. Since our scenarios consist of single clients the stochastic and demand-robust versions of the $k$-center problem we solve correspond to $k$-median and $k$-center respectively.

Our last theorem shows that \emax  generalizes the stochastic and demand-robust models as well as a \hybrid model which smoothly interpolates between stochastic and demand-robust optimization.

\begin{restatable}{theorem}{RedToEmax}\label{thm:RedToEmax}
	An  $\alpha$-approximation for a two-stage covering algorithm in the \emax model implies an  $\alpha$-approximation for the corresponding two-stage covering problem in the stochastic, demand-robust, and  \hybrid models.
\end{restatable}

\noindent For cleanliness of exposition, we defer a formal definition and discussion of the \hybrid model as well as the intuition and proof for Theorem \ref{thm:RedToEmax} to Section \ref{sec:RedToEmax}. As a corollary of Theorems \ref{thm:Applications} and \ref{thm:RedToEmax},  we immediately recover polynomial-time approximations for \hybrid MST, UFL, Steiner tree and min-cut.\footnote{Though not $k$-center since its cost function is not linear.}\\

\section{Reducing \emax to \trunc}\label{sec:trunc}
%Reasoning about the expected max of a given solution is challenging: whether a solution pays for any particular scenario depends on which other scenarios appear. Thus, even computing the cost of a solution for  \emax  seems to require inspecting  exponentially many different ways in which scenarios appear. 

In this section, we demonstrate a technique to simplify  both computing and reasoning about $\ce$ by reducing a \emax problem to a \trunc problem with only a small loss in the approximation factor. Specifically, we show the following theorem.
\EmaxToCostB*
\noindent As earlier noted, we show this by observing that a set of expensive scenarios with ``large'' total probability mass dominates the cost of a given \emax solution. 

We begin by observing that the expected max of a set of independent random variables is approximately bounded by the most expensive of these random variables whose probabilities sum to $1$. We remark that this result can be seen to follow from results regarding the ``correlation gap'' \cite{ADSY-OR12,Alaei-SICOMP14} which show a similar bound where instead of $\max$ we have any sub-modular function. We give a different proof in \S\ref{sec:defProofsTrunc} for completeness that we find simpler in our setting where we consider the $\max$ and not any sub-modular function.

%See, for example, \cite{ADSY-OR12}, which shows that a similar bound holds if instead of $\max$ we have any submodular objective.  both for the sake of completeness and because we find the proof for our specific case of max simpler. We give our own proof of this lemma in \S\ref{sec:defProofsTrunc} both for the sake of completeness and because we find the proof instru

\begin{restatable}{lemma}{eTrick}
\label{lem:eTrick}
Let $\pmb{Y} = \{Y_1, \ldots, Y_m\}$ be a set of independent Bernoulli r.v.s, where $Y_s$ is $1$ with probability $p_s$, and $0$ otherwise. Let $v_s \in \reals_{\geq 0}$ be a value associated with $Y_s$. WLOG assume $v_{s} \geq v_{s+1}$ for $s \in [m-1]$. Let $b = \min \{ a: \sum_{s=1}^a
p_s \ge 1 \}$. 
%Let $A \subseteq [m]$ be a random set containing $s$ iff $Y_s = 1$. 
Then 
\begin{align*}
\Big(1 - \frac{1}{e}\Big) \Big(v_b + \sum_{s}p_s \cdot (v_s-v_b)^+ \Big) \leq \E_{\pmb{Y}} \Big[\max_{s }\{Y_s \cdot v_s \} \Big] \leq v_b + \sum_{s}p_s \cdot (v_s - v_b)^+,
\end{align*}
where $x^+ := \max\{x, 0\}$.
\end{restatable}
%\noindent We will apply this lemma to a solution $(X_1, \pmb{X_2})$ where $v_b$ will be $B(X_1,\pmb{X_2})$ and $\bigcup_{s \leq b} s$ will be $M(X_1, \pmb{X_2})$.
\noindent For a \emph{given solution} $(X_1, \pmb{X}_2)$ to \emax, \Cref{lem:eTrick}  yields a computationally tractable form of $\ce$. Specifically, let our scenarios be indexed such that  $\cost(X_1, X_2^{(s)}) \geq \cost(X_1, X_2^{(s+1)})$ and let $b$ be the smallest positive integer such that $\sum_{s=1}^b p_s \geq 1$. 
We define the following terms   analogous to those in the lemma (see \Cref{fig:MAndB} for an illustration):
\begin{align}\label{defn:MAndB}
M(X_1, \pmb{X}_2) := [b] \qquad  \text{ and } \qquad B(X_1, \pmb{X}_2) := \cost(X_1, X_2^{(b)} ).
\end{align}
Notice that $\sum_{s \in M(X_1, \pmb{X_2})} p_s < 2$. Now, by letting $B(X_1, \pmb{X_2})$ be $v_b$ in \Cref{lem:eTrick}, we can approximate $\ce(X_1, \pmb{X_2})$. 
%might seem that up to a $\frac{1-1/e}{2}$ factor we have an efficiently-computable approximation for $\ce$. 
However, we would like to estimate $\ce(X_1, \pmb{X}_2)$ within an LP where $(X_1, \pmb{X}_2)$ are variables since our algorithms are LP based. Unfortunately, it is not clear how to capture $v_b$ in an LP and so it is not clear how to directly use \Cref{lem:eTrick} to estimate $\ce(X_1, \pmb{X}_2)$ within an LP.
%However, our methods in \S\ref{sec:subaddAppl} will be LP-based it is not clear how to capture $v_b$ in an LP.

\begin{figure} 
\begin{tikzpicture}[scale=1.45]
\draw [line width = 5] (0, 0)--(10,0);

\draw [dashed](-.3, 0.8)--(3.7, 0.8); \draw [dashed](-.3, 0.8)--(-.3, .5); \draw [dashed](3.7, 0.8)--(3.7, .5); \node[above] at (1.5, 0.8) {$M(X_1, \pmb{X_2})$};

\draw [fill = red!30] (0,0) circle [radius=.3]; \node at (0,0) {.4}; \node [below] at (0, -.3) {\scriptsize $\cost(X_1, X_2^{(1)}) \geq$};
\draw [fill = red!30] (1.5,0) circle [radius=.3]; \node at (1.5,0) {.4}; \node [below] at (1.5, -.3) {\scriptsize $\cost(X_1, X_2^{(2)}) \geq$};
\draw [fill = red!30] (3.4,0) circle [radius=.3]; \node at (3.4,0) {.3}; \node [below] at (3.4, -.3) {\scriptsize $\cost(X_1, X_2^{(3)}) = B(X_1, \pmb{X_2}) \geq$};

\draw [fill = green!30] (5.25,0) circle [radius=.3]; \node at (5.25,0) {.6}; \node [below] at (5.25, -.3) {\scriptsize $\cost(X_1, X_2^{(4)}) \geq$};
\draw [fill = green!30] (6.75,0) circle [radius=.3]; \node at (6.75,0) {.7};\node [below] at (6.75, -.3) {\scriptsize $\cost(X_1, X_2^{(5)}) \geq$};
\draw [fill = green!30] (8.25,0) circle [radius=.3]; \node at (8.25,0) {1}; \node [below] at (8.25, -.3) {\scriptsize $\cost(X_1, X_2^{(6)}) \geq$};
\draw [fill = green!30] (10,0) circle [radius=.3]; \node at (10,0) {.2};\node [below] at (10, -.3) {\scriptsize $\cost(X_1, X_2^{(7)}) $};
\end{tikzpicture}
\caption{$M(X_1, \pmb{X_2})$ and $B(X_1, \pmb{X_2})$. Red circles: scenarios in $M(X_1, \pmb{X_2})$. Green circles: all other scenarios. Numbers in circles: probabilities. Scenarios arranged left to right in descending order of $\cost(X_1, X_2^{(s)} )$.}
\label{fig:MAndB}
\end{figure}
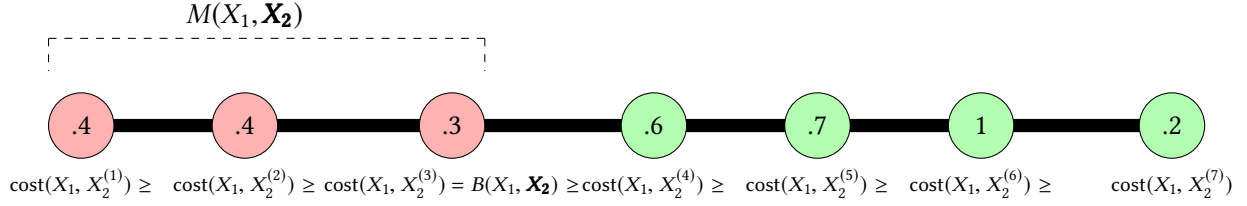

For this reason, we derive an even simpler form of the above approximation of the expected max which can be computed using an LP. In particular, we show that the expected max is approximately the $\ct$ objective. We remind the reader that, as per Eq.\eqref{eq:truncObjective}, $\ct(X_1, \pmb{X_2}):=\min_B [ B + \sum_{s \in [m]} p_s\cdot (\cost(X_1,X_2^{(s)}) - B)^+]$. The following lemma shows that the $B$ achieving the minimum in $\ct(X_1, \pmb{X_2})$ is $B(X_1, \pmb{X_2})$ and therefore shows that $\ct$ is a good approximation of $\ce$.

%Ravi: Before lemma 3.2, it would be helpful to remind the reader to look at Eqn (5) defining cost_Trunc. Then say something like Lemma 3.2. shows that the value B achieving the minimum in Equation (5) is the same as B(X_1,X_2) defined above. This clarifies that the lemma relates these two different definitions of cost_Trunc.

\begin{restatable}{lemma}{altTruncForm}\label{lem:altTruncForm}
Let $(X_1, \pmb{X_2})$ be a solution to a \trunc or \emax problem. We have
\begin{align*}
%&\ct(X_1, \pmb{X_2}) = B(X_1, \pmb{X_2}) + \sum_s {p_s} \Big(\cost(X_1, X_2^{(s)} - B(X_1, \pmb{X_2}) \Big)^+ \qquad \text{and}\\
&B(X_1, \pmb{X_2}) = \argmin_B \Big[ B + \sum_{s \in [m]} p_s \cdot (\cost(X_1, X_2^{(s)})-B)^+ \Big],
\end{align*}
where the $\argmin$ takes the largest $B$ minimizing the relevant quantity.
\end{restatable}
\gdef\convexityLemmaProof{
\begin{proof}
To clear our notation we let $\bar{B} := B(X_1, \pmb{X_2})$, $c_s : = \cost(X_1, X_2^{(s)})$ and $\bar{M} := M(X_1, \pmb{X_2})$. Let $f(B) := B + \sum_{s \in [m]} p_s \cdot (c_s-B)^+$. We argue that $\bar{B}$ is the largest global minimum of $f$ by showing that for any $\eps > 0$ we know that $f(\bar{B}) < f(\bar{B} + \eps)$ and $f(\bar{B}) \leq f(\bar{B} - \eps)$. 

We begin by noting that for any reals $a \leq b$ we have
\begin{align}
a^+ - b^+ \geq a-b \label{eq:truncAdditive}
\end{align}
by casing on which of $a$ and $b$ are larger than $0$.

%We first show that $f(B)$ is strictly convex in $B$. We begin by noticing that by a simple case analysis it holds that
%\begin{align}
%(a + b)^+ \leq a^+ + b^+. \label{eq:truncAdditive}
%\end{align}
%Let $\theta \in (0,1)$ be fixed and arbitrary. Now consider $f(\theta B_1 + (1 - \theta) B_2)$ for fixed arbitrary $B_1 \neq B_2$.   We have
%\begin{align*}
%& f(\theta B_1 + (1 - \theta) B_2) \\
%&= \theta B_1 + (1 - \theta) B_2 + \sum_{s} p_s \cdot (c_s - \theta B_1 - (1 - \theta) B_2)^+\\
%&= \theta B_1 + (1 - \theta) B_2 + \sum_{s} p_s \cdot \left(\theta (c_s - B_1) + (1-\theta)(c_s - B_2) \right)^+\\
%&\leq \theta B_1 + (1 - \theta) B_2 + \sum_{s} p_s \cdot \left( \theta \left(c_s - B_1 \right)^+ +  (1 - \theta)\left(c_s -  B_2 \right)^+ \right) \bec{\Cref{eq:truncAdditive}}\\
%& = \theta \left( B_1 + \sum_{s} p_s \cdot  \left(c_s - B_1 \right)^+ \right) + (1 - \theta)  \left(B_2 + \sum_{s} p_s \cdot \left(c_s - B_2 \right)^+ \right) \\
%&= \theta f(B_1) + (1-\theta)f(B_2).
%\end{align*}
%Thus, $f(\theta B_1 + (1 - \theta) B_2) \leq \theta f(B_1) + (1-\theta)f(B_2)$ and so $f$ is convex in $B$; it follows that any local minimum of $f$ is a global minimum. To show our lemma, then, we need only show that $\bar{B}$ is the largest local minimum.

Let $\hat{M} := \{s \in \bar{M} : c_s > \bar{B}\}$. Notice that $\sum_{s \in \hat{M}} p_s < 1$. For fixed and arbitrary $\eps > 0$ consider the relative values of $f(\bar{B})$ and $f(\bar{B} + \eps)$. We have
\begin{align}
f(\bar{B}+ \eps) - f(\bar{B}) &= \eps + \sum_{s} p_s \cdot \left((c_s - \bar{B} - \eps)^+ - (c_s - \bar{B})^+ \right) \nonumber \\
&= \eps + \sum_{s \in \hat{M}} p_s \cdot \left((c_s - \bar{B} - \eps)^+ - (c_s - \bar{B})^+ \right), \label{eq:useHatM}
\end{align}
where \eqref{eq:useHatM} follows since for $s \not \in \hat{M}$ we have $c_s \leq \bar{B}$ and so $\left((c_s - \bar{B} - \eps)^+ - (c_s - \bar{B})^+ \right) = 0$ for  $s \not \in \hat{M}$. Now noticing that for every $s$ we have $(c_s - \bar{B} - \eps) \leq (c_s - \bar{B})$, applying \eqref{eq:truncAdditive} to \eqref{eq:useHatM} gives
\begin{align*}
f(\bar{B}+ \eps) - f(\bar{B})  &\geq  \eps + \sum_{s \in \hat{M}} p_s \cdot \left(-\eps\right) 
 = \eps \left(1 - \sum_{s \in \hat{M}} p_s \right)
 > 0, 
\end{align*}
where the last inequality uses $\sum_{s \in \hat{M}} p_s < 1$. Thus, we have $f(\bar{B} + \eps) > f(\bar{B})$.

Now consider the relative values of $f(\bar{B})$ and $f(\bar{B} - \eps)$. We have 
\begin{align*}
f(\bar{B} - \eps) - f(\bar{B})  &= -\eps + \sum_s p_s \cdot \left((c_s - \bar{B} + \eps)^+ - (c_s - \bar{B} )^+ \right) \nonumber \\
 &\geq -\eps + \sum_{s \in \bar{M}} p_s \cdot \left((c_s - \bar{B} + \eps)^+ - (c_s - \bar{B} )^+ \right) \bec{$(c_s - \bar{B} + \eps)^+ \geq (c_s - \bar{B} )^+$} \\
  &\geq -\eps + \sum_{s \in \bar{M}} p_s \cdot \left((c_s - \bar{B} + \eps) - (c_s - \bar{B} ) \right) \bec{$c_s \geq \bar{B}$ for $s \in \bar{M}$}\nonumber\\
&\geq \eps \Big(1 - \sum_{s \in \bar{M}} p_s \Big) \quad \geq \quad 0 \bec{$\sum_{s \in \bar{M}} p_s \geq 1$}.
\end{align*}
%where \Cref{eq:useHatMAgain} follows because $c_s \leq \bar{B}$ for $s \not \in \bar{M}$ and $\eps$ was chosen to be sufficiently small so that if $c_s \leq $. Now noticing that for every $s$ we have $(c_s - \bar{B}) \leq (c_s - \bar{B} + \eps)$, applying \Cref{eq:truncAdditive} gives
%\begin{align*}
%f(\bar{B}) - f(\bar{B} - \eps) &\geq\eps + \sum_{s \in \bar{M}} p_s \cdot \left(-\eps \right) \bec{TODO} \\
%&\geq \eps \left(1 - \sum_{s \in \bar{M}} p_s \right) \bec{$\bar{M}$ dfn.} \\
%&\geq 0
%\end{align*}
Thus, for any $\eps > 0$ we know that $f(\bar{B}) < f(\bar{B} + \eps)$ and $f(\bar{B}) \leq f(\bar{B} - \eps)$. It follows that, not only is $\bar{B}$ a global minimum of $f$ but it is the largest global minimum. The lemma follows immediately.
\end{proof}
}

\ifFULL
\convexityLemmaProof
\else
\begin{proof}[Proof Sketch]
The rough idea of the proof is to show that $B + \sum_s {p_s} (\cost(X_1, X_2^{(s)} -B )^+$ is convex in $B$ and that $B(X_1, \pmb{X_2})$ is a local minimum. In particular, imagine that $B$ is currently set at $B(X_1, \pmb{X_2})$ and consider what happens to $B + \sum_s {p_s} (\cost(X_1, X_2^{(s)} -B )^+$ if we shift $B$ to be smaller. Recall that we have at least one probability mass across elements which are larger than $B$ by definition of $B(X_1, \pmb{X_2})$. Thus, when we shift $B$ to be smaller, $B$ decreases slower than $\sum_s {p_s} (\cost(X_1, X_2^{(s)} - B)^+$ increases and so $B + \sum_s {p_s} (\cost(X_1, X_2^{(s)} -B )^+$ becomes larger overall. The case when $B$ is made larger is symmetric. The full proof is available in \S\ref{sec:defProofsTrunc}.
\end{proof}
\fi
Using \Cref{lem:eTrick} and \Cref{lem:altTruncForm}, it is easy to show the following two lemmas. These lemmas---proved in \S\ref{sec:defProofsTrunc}---upper and lower bound the \emax cost of a solution with respect to its \trunc solution respectively.

\begin{restatable}{lemma}{xEToTrunc}\label{lem:xEToTrunc}
For feasible solution $(X_1, \pmb{X}_2)$ of any $\Pe$ we have, $\ce(X_1, \pmb{X}_2) \leq \ct(X_1, \pmb{X}_2)$.
\end{restatable}

\begin{restatable}{lemma}{TOptToEOpt}\label{lem:TOptToEOpt}
Let $\Pe$ be a \emax problem and $\Pt$ be its truncated version. Let $(E_1, \pmb{E_2})$ and $(T_1, \pmb{T_2})$ be optimal solutions to $\Pe$ and $\Pt$ respectively. We have $\ct(T_1, \pmb{T_2}) \leq \left( \frac{1}{1 - 1/e} \right )\ce(E_1, \pmb{E_2})$.
\end{restatable}

The preceding lemmas allow us to conclude that an $\alpha$-approximation algorithm for a \trunc problem is an $O(\alpha)$-approximation algorithm for the corresponding \emax problem.
\begin{proof}[Proof of \Cref{thm:EmaxToCostB}]
Let $(\hat{T}_1, \pmb{\hat{T}_2})$ be the solution returned by an $\alpha$-approximation algorithm for $\Pt$. Let $(E_1, \pmb{E_2})$ and $(T_1, \pmb{T_2})$ be the optimal solutions to $\Pe$ and $\Pt$ respectively. 
\ifFULL
\begin{align*}
\ce(\hat{T}_1, \pmb{\hat{T}_2}) &\leq \ct(\hat{T}_1, \pmb{\hat{T}_2}) \bec{\Cref{lem:xEToTrunc}}\\
&\leq \alpha \cdot \ct(T_1, \pmb{T_2}) \bec{$(\hat{T}_1, \pmb{\hat{T}_2})$ is  an $\alpha$ approx.}\\
&\leq \left(\frac{2 \alpha}{1-1/e}\right) \ce(E_1, \pmb{E_2})  \bec{\Cref{lem:TOptToEOpt}}.
\end{align*}
\else
By \Cref{lem:xEToTrunc} we have $\ce(\hat{T}_1, \pmb{\hat{T}_2}) \leq \ct(\hat{T}_1, \pmb{\hat{T}_2})$. Since $(\hat{T}_1, \pmb{\hat{T}_2})$ is an $\alpha$-approximation we have this is at most $\alpha \cdot \ct(T_1, \pmb{T_2})$. Applying \Cref{lem:TOptToEOpt} this is at most $\left(\frac{ \alpha}{1-1/e}\right) \ce(E_1, \pmb{E_2})$.
\fi
Since any solution that is feasible for $\Pt$ is also feasible for $\Pe$, we conclude that $(\hat{T}_1, \pmb{\hat{T}_2})$ is a feasible solution for $\Pe$ with cost in $\Pe$ at most $\left( \frac{\alpha}{1-1/e}  \right) \ce(E_1, \pmb{E_2})$, giving our theorem.
\end{proof}

% !TeX root = arXiv.tex
% !TEX root = arXiv.tex

%%%%%%%%%%%%%%%%%%%%%%%%%%%
\section{Applications to Linear Two-Stage Covering Problems} \label{sec:subaddAppl}

In this section we give an $O(\log n + \log m)$-approximation algorithm for \emax MST and  $O(1)$ approximation algorithms for \emax Steiner tree, \emax facility location, and \emax min-cut.
Our algorithms are LP based. To derive our algorithms we use our reduction from \S\ref{sec:trunc} to transform a \emax problem into a \trunc problem with only a small constant loss in the approximation factor. This transformation allows us to adapt existing LP rounding techniques in which every scenario has a rounding cost close to its fractional cost~\cite{RS-IPCO04,GPRS-STOC04,SS2006} to solve our \trunc problems and, therefore, our \emax problems.

We first give two general techniques to solve a \trunc problem.

%We attain these results by using  our reduction from \S\ref{sec:trunc} that transforms a \emax problem into a \trunc problem at the cost of small constants in the approximation factor. Specifically, we show that \trunc problems can be represented as LPs which allows us to adapt LP rounding techniques in which every scenario has a rounding cost close to its fractional cost; thereby solving \trunc and \emax problems. We begin this section by expounding upon these techniques.

%---------------------------------------------------------------------------

\subsection{General Techniques}\label{subsec:genTruncTechs}

Our \emph{first technique} is to represent $\ct$ as an LP objective.  For this technique we need to extend the definition of $\ct$ from an integral solution $(X_1, \pmb{X_2})$ to a fractional solution $(x_1, \pmb{x_2})$. To do so, in each of our problems we locally define $\cost(x_1, x_2^{(s)})$ for fractional solution $(x_1, x_2^{(s)})$  to scenario $s$ and let $\ct(x_1, \pmb{x_2})$ be defined similarly to the integral case, i.e. for fractional $(x_1, \pmb{x_2})$,
\begin{align}
\ct(x_1, \pmb{x_2}) := \min_B \Big[B + \sum_s p_s (\cost(x_1, x_2{(s)}) - B)^+ \Big].\label{eq:fractionalCT}
\end{align}

Given a minimization LP, it is easy to see that by introducing an additional variable to represent $B$ and additional variables to represent $(\cost(x_1, x_2{(s)}) - B)^+$ for every $s$,  we can represent $\ct(x_1, \pmb{x_2})$ in an LP.  For cleanliness of exposition, when we write our LPs we omit these additional variables and simply write our objective as ``$\ct(x_1, \pmb{x_2})$.'' Moreover,  even though some of our LPs have an exponential number of constraints, we rely on the existence of efficient separation oracles for these LPs. It is easy to verify that this holds even after one introduces the additional variables needed to represent $\ct(x_1, \pmb{x_2})$.

We also extend  $M$ and $B$ from the integral case as defined in \S\ref{sec:trunc} to the fractional case in the following natural way. Given a fractional solution $(x_1, \pmb{x_2})$ and a cost function on fractional solutions, $\cost$, WLOG let our scenarios be indexed such that  $\cost(x_1, x_2^{(s)}) \geq \cost(x_1, x_2^{(s+1)})$. Let $b$ be the smallest positive integer such that $\sum_{s=1}^b p_s \geq 1$. For fractional $(x_1, \pmb{x_2})$, we define
\begin{align}
&M(x_1, \pmb{x}_2) := [b]\\
&B(x_1, \pmb{x}_2) := \min_{s \in M(x_1, \pmb{x}_2)} \cost(x_1, x_2^{(s)} ) \label{eq:FracB}.
\end{align}

\begin{rem} \label{rem:lem:altTruncForm}
 It is easy to verify that the proof of \Cref{lem:altTruncForm} also holds for $\ct(x_1, \pmb{x_2})$  for fractional $(x_1, \pmb{x_2})$. We will therefore  invoke it on fractional $(x_1, \pmb{x_2})$, even though it is stated only for integral $(X_1, \pmb{X_2})$.
\end{rem}

Our \emph{second technique} is a generic {rounding technique}  for  \trunc problems. Several past works in two-stage optimization show that it is possible to round an LP solution such that the resulting integral solution has cost roughly the same as the fractional solution for \emph{every scenario}. We prove the following lemma to make use of such rounding algorithms.
%This lemma will also be useful in proving certain structural properties regarding the optimal solution for \trunc MST and \trunc min-cut.

\begin{lemma}\label{lem:scenByScen}
Let $\Pt$ be a \trunc problem. Let $(X_1, \pmb{X_2})$ and $(Y_1, \pmb{Y_2})$ be integral or fractional solutions to $\Pt$. If for every scenario $s$ we have $\cost(X_1,X_{2}^{(s)}) \leq c \cdot \cost(Y_1, Y_2^{(s)})$ then
\begin{align*}
\ct(X_1, \pmb{X_2}) \leq c \cdot \ct(Y_1, \pmb{Y_2}).
\end{align*}
\end{lemma}
\begin{proof}
We have
\begin{align*}
& \quad \ct(X_1, \pmb{X_2}) = \min_B \left[ B + \sum_s p_s \cdot (\cost(X_1, X_2^{(s)})-B)^+)\right] \\
& \leq  c \cdot B(Y_1, \pmb{Y_2}) + \sum_s p_s \cdot (\cost(X_1, X_2^{(s)})-c \cdot B(Y_1, \pmb{Y_2}))^+) \bec{letting $B = c \cdot B(Y_1, \pmb{Y_2})$} \\
& \leq c \cdot B(Y_1, \pmb{Y_2}) + \sum_s p_s \cdot (c \cdot \cost(Y_1, Y_2^{(s)})- c \cdot B(Y_1, \pmb{Y_2}))^+) \bec{$\cost(X_1,X_{2}^{(s)}) \leq c \cdot \cost(Y_1, Y_2^{(s)})$}\\
& = c \cdot \left( B(Y_1, \pmb{Y_2}) + \sum_s p_s \cdot ( \cost(Y_{1}, Y_{2}^{(s)})- B(Y_1, \pmb{Y_2}))^+) \right)\\
& = c \cdot  \ct(Y_1, \pmb{Y_2})\bec{\Cref{lem:altTruncForm}}.
\end{align*}
\end{proof}

%------------------------------------------------------------
\subsection{Uncapacitated Facility Location}\label{sec:FL}
In this section we give a polynomial-time $\left( \frac{8}{1-1/e} \right)$-approximation algorithm for \emax uncapacitated facility location (UFL).

\begin{defn}[\emax UFL]
We are given a set of facilities $F$ and a set of clients $\met$ with a metric $c_{ij}$ specifying the distances between every client $j$ and facility $i$. We are also given scenarios $S_1, \ldots, S_m \subseteq \met$, where in scenario $S_s$ client $j$ has demand $d_j^s \in \{0, 1\}$\footnote{This easily generalizes to more demand.}, and a probability $p_s$ for each scenario. Facility $i$'s opening cost is $f_{1,i}$  in the first stage and $f_{2,i}^{(s)}$  in scenario $S_s$. These opening costs can be $\infty$, which indicates  the facility cannot be opened. A feasible solution consists of a set of first and second stage facilities $(X_1, \pmb{X_2})$ s.t. $X_1\cup \bigcup_s X_2^{(s)} \neq \emptyset$. The cost for scenario $s$ in  solution $(X_1, \pmb{X_2})$ is
\begin{align*}
\cost(X_1, X_2^{(s)}) := \sum_{i \in X_1} f_{1,i} + \sum_{i \in X_2^{(s)}} f_{2,i}^{(s)} + \sum_{j \in S_s} \min_{i \in X_1 \cup X_2^{(s)}} c_{ij}.
\end{align*}
 The total cost of our solution $(X_1, \pmb{X_2})$ is $\ce(X_1, \pmb{X_2}) := \E_{A \sim \pmb{p}} \left[ \max_{s \in A} \{\cost(X_1, X_2^{(s)}) \} \right]$.
 \end{defn}

 Our algorithm is based on the work of Ravi and Sinha~\cite{RS-IPCO04} on two-stage stochastic UFL. This work shows how to round an LP such that every scenario has a ``good" cost after rounding. Applying \Cref{lem:scenByScen} to this rounding gives an algorithm that approximates \trunc UFL, which by \Cref{thm:EmaxToCostB} is sufficient to approximate \emax UFL.

We use the following LP. Variable $z_{ij}^{(s)}$ corresponds to whether client $j$ is served by facility $i$ in scenario $s$. Variables $x_1(i)$ and $x_2^{(s)}(i)$ corresponds to whether facility $i$ is opened in the first stage or scenario $s$, respectively. For a fractional solution $(x_1, \pmb{x_2})$, we define
\begin{align*}
\cost(x_1, x_2^{(s)}) := \sum_{i \in F} \Big[ x_1(i) \cdot f_{1, i} + x_2^{(s)}(i) \cdot f_{2, i}^{(s)} + \sum_{j \in \met} \hat{z}_{ij}^{(s)}\cdot c_{ij} \Big],
\end{align*}
 where  $\hat{z}_{ij}^{(s)}$ is the natural fractional assignment given fractional facilities $(x_1, x_2^{(s)})$; namely, one that sends clients to their nearest fractionally opened facilities.  As described by Eq.\eqref{eq:fractionalCT}, this definition of $\cost(x_1, x_2^{(s)})$ defines $\ct(x_1, \pmb{x}_2)$ for fractional $(x_1, \pmb{x_2})$, which allows us to define our LP.
\begin{align}
\label{LP:FLLP}
\min \qquad &\ct(x_1, \pmb{x_2}) \tag{UFL LP}\\
\text{s.t.} \qquad & \sum_{i \in F}z_{ij}^{(s)} \geq d_j^{(s)} & \forall j \in \met, \forall s \notag \\
&z_{ij}^{(s)} \leq x_1(i) + x_2^{(s)}(i) & \forall i \in F, \forall  j \in \met, \forall s \notag \\
&0 \leq x_1, \pmb{x_2}, \pmb{z}	\notag
\end{align}

\noindent Note that an integral solution to the above LP  is a feasible solution for \emax UFL. Ravi and Sinha showed how to round this LP.

\begin{lemma}[Theorem 2, Lemma 1 in \cite{RS-IPCO04}]\label{lem:roundFLLP}
Given a fractional solution $(x_1, \pmb{x_2})$ to \ref{LP:FLLP}, it is possible to round it  to integral $(X_1, \pmb{X_2})$ in polynomial time s.t.\ for every scenario $s$ we have $\cost(X_1, X_2^{(s)}) \leq 8 \cdot \cost(x_1, x_2^{(s)})$.
\end{lemma}
We now give our approximation algorithm for \emax UFL.

\begin{theorem}
\emax UFL can be $\left( \frac{8}{1-1/e} \right)$-approximated in polynomial time.
\end{theorem}
\begin{proof}
Our algorithm starts by solving \ref{LP:FLLP} to get a fractional $(x_1, \pmb{x_2})$. Next, round $(x_1, \pmb{x_2})$ using \Cref{lem:roundFLLP} to  integral $(X_1, \pmb{X_2})$. Return $(X_1, \pmb{X_2})$.

Let $(O_1, \pmb{O_2})$ be the optimal integral solution to the \trunc instance of our problem and let $(o_1, \pmb{o_2})$ be its corresponding characteristic function. By definition, $\ct(o_1, \pmb{o_2}) = \ct(O_1, \pmb{O_2})$.
Now using \Cref{lem:scenByScen} and \Cref{lem:roundFLLP} it follows that
\begin{align*}
\ct(X_1, X_2) &\leq 8 \cdot \ct(x_1, \pmb{x_2}).
\end{align*}
Since $(o_1, \pmb{o_2})$ feasible for \ref{LP:FLLP}, we get
\[	\ct(X_1, X_2) \quad \leq \quad 8 \cdot \ct(o_1, \pmb{o_2}) \quad = \quad 8 \cdot \ct(O_1, \pmb{O_2}).
\]

\noindent Thus, our algorithm is an $8$-approximation for \trunc UFL. Applying \Cref{thm:EmaxToCostB} gives a $\left( \frac{8}{1-1/e} \right)$-approximation for \emax UFL.

Lastly, notice that our algorithm is trivially polynomial-time.
\end{proof}

%------------------------------------------------------------
\subsection{Steiner Tree}\label{sec:ST}
In this section we give a $\left(\frac{30}{1-1/e}\right)$-approximation for \emax rooted Steiner tree.

\begin{defn}[\emax Rooted Steiner tree]
We are given a graph $G = (V,E)$, a root $r \in V$, a cost $c_e$ for each edge $e$. We are also given scenarios $S_1, \ldots, S_m \subseteq V$, each with an associated probability $p_s$ and an \emph{inflation factor} $\sigma_s > 0$. We must find a first stage solution $X_1 \subseteq E$ and a second-stage solution for every scenario, $X_2^{(j)} \subseteq E$. A  solution is feasible if for every $s$ we have $X_1 \cup X_2^{(s)}$ connects $\{r\} \cup S_s$. The cost  for scenario $s$ in  solution $(X_1, \pmb{X_2})$  is
\begin{align}
\cost(X_1, X_2^{(s)}) := \sum_{e \in X_1} c_e + \sigma_s \cdot \sum_{e \in X_2^{(s)}} c_e.
\end{align}
The total cost we pay for solution  $(X_1, \pmb{X_2})$ is $\ce(X_1, \pmb{X_2}) := E_{A \sim \pmb{p}} \left[ \max_{s \in A} \{ \cost(X_1, X_2^{(s)}) \} \right]$.
\end{defn}

Our algorithm is based on an LP rounding algorithm  of Gupta et al.\ \cite{GRS-FOCS04} for two-stage stochastic Steiner tree. Roughly, we use \Cref{lem:scenByScen} to argue that the first stage solution for every optimal \trunc solution is, up to small constants, a tree rooted at  $r$.
This structural property allows us to write an LP that approximately captures \trunc Steiner tree. Gupta et al.\ \cite{GRS-FOCS04} showed that this LP can be rounded s.t. every scenario has a good cost. As in the previous section, we combine this rounding with \Cref{lem:scenByScen} to derive an approximation algorithm for  \trunc  Steiner tree, which is sufficient for approximating \emax Steiner tree by \Cref{thm:EmaxToCostB}.

We begin by arguing that up to small constants, the optimal first stage solution is a tree rooted at $r$.

\begin{lemma} \label{lem:rootTreeST}
There exists an integral solution $(\hat{X}_1, \pmb{\hat{X}_2})$ to \trunc Steiner tree s.t.\ $G[\hat{X}_1]$ is a tree rooted at $r$ and $\ct(\hat{X}_1, \pmb{\hat{X}_2}) \leq 2 \cdot \ct(O_1, \pmb{O_2})$, where $(O_1, \pmb{O_2})$ is the optimal solution to \trunc Steiner tree.
\end{lemma}
\begin{proof}
Lemma 4.1 of Dhamdhere et al.\ \cite{DGRS-FOCS05} shows that given $(O_1, \pmb{O_2})$ it is possible to modify it to a feasible solution $(\hat{X}_1, \pmb{\hat{X}_2})$ such that $G[\hat{X}_1]$ is a tree rooted at $r$ and $\cost(\hat{X}_1, \hat{X}_2^{(s)}) \leq 2 \cdot \cost(O_1, O_2^{(s)})$ for every $s$. It follows by \Cref{lem:scenByScen} that $\ct(\hat{X}_1, \pmb{\hat{X}_2}) \leq 2 \cdot \ct(O_1, \pmb{O_2})$.
\end{proof}

We now describe how to formulate an LP that leverages the structural property in \Cref{lem:rootTreeST}. In particular, this indicates that as one gets closer to $r$, one must fractionally buy edges to a greater and greater extent. This constraint can be captured in an LP. Specifically, every node in a scenario (a.k.a.\ terminal) is the source of one unit of flow that is ultimately routed to $r$; this flow follows a path whose fractional ``first stage-ness" is monotonically increasing.

More formally, we copy each edge $e = \{u,v\}$ into two directed edges $(u,v)$ and $(v,u)$. Let $\vec{e}$ be either one of these directed edges. Next, for each such directed edge $\vec{e}$ and every terminal in $t \in \bigcup_{s} S_s$, we define variables $r_{1}(t, \vec{e})$ and $r_{2}^{(s)}(t, \vec{e})$ for every $s$ to represent how much $t$ is connected to $r$ by $e$ in the first stage and in  scenario $s$, respectively. Also, for undirected edge $e$, define variables $x_1(e)$ and $x_{2}^{(s)}(e)$ to stand for how much we buy $e$ in the first stage and scenario $s$, respectively. For fractional $(x_1, \pmb{x_2})$, we define
\begin{align*}
\ct(x_1, x_2^{(s)}):= \sum_{e} c_e \cdot x_1(e) + \sigma_s \cdot c_e  \cdot x_2(e),
\end{align*}
 which as described by Eq.\eqref{eq:fractionalCT} also defines $\ct(x_1, \pmb{x}_2)$. Letting $\delta^-(v)$ and $\delta^+(v)$ stand for all directed edges going into and out of $v$, respectively. The following is our LP.
\begin{align}
\label{LP:STEmaxLP}
\min \qquad  &\ct(x_1, \pmb{x_2}) \tag{ST LP}\\
\text{s.t.} \qquad & \sum_{\vec{e} \in \delta^+(v)}r_{1}(t, \vec{e}) + r_{2}^{(s)}(t, \vec{e}) = \sum_{\vec{e} \in \delta^-(v)}r_{1}(t, \vec{e}) + r_{2}^{(s)}(t, \vec{e}) & \forall s, t \in S_s, v \not \in \{t, r\} \notag \\
&\left[\sum_{\vec{e} \in \delta^+(t)} r_{1}(t, \vec{e}) + r^{(s)}_{2}(t, \vec{e}) \right] - \left[ \sum_{\vec{e} \in \delta^-(t)} r_{1}(t, \vec{e}) + r^{(s)}_{2}(t, \vec{e}) \right] \geq 1 &\forall s, t \in S_s \notag\\
& \sum_{\vec{e} \in \delta^-(v)} r_{1}(t, \vec{e}) \leq \sum_{\vec{e} \in \delta^+(v)} r_{1}(t, \vec{e}) &\forall s, t \in S_s, v \not \in \{t, r\} \label{line:LPMono}\\
&r_{1}(t, \vec{e}) \leq x_{1}(e); r_{2}^{(s)}(t, \vec{e}) \leq x_{2}^{(s)}(e)  &  \forall s, t \in S_s, \vec{e} \notag\\
&r, x_1, \pmb{x_2} \geq 0	\notag
\end{align}

Notably, Eq.~\eqref{line:LPMono} enforces that terminal $t$ is serviced by the first stage more and more as one moves closer to the root. The characteristic vector of $(\hat{X}_1, \pmb{\hat{X}_2})$ as described in \Cref{lem:rootTreeST} gives a feasible solution to \ref{LP:STEmaxLP}. As a result, \Cref{lem:rootTreeST} demonstrates that \ref{LP:STEmaxLP} has nearly optimal objective as stated in the following corollary.

\begin{corollary}\label{cor:STCor}
Let $(x_1, \pmb{x_2})$ be  the optimal solution of \ref{LP:STEmaxLP}. We have $\ct(x_1, \pmb{x_2}) \leq 2 \cdot \ct(O_1, \pmb{O_2})$, where $(O_1, \pmb{O_2})$ is the optimal solution to \trunc Steiner tree.
\end{corollary}
\begin{proof}
Let $(\hat{x}_1, \pmb{\hat{x}_2})$ be the characteristic vector of $(\hat{X}_1, \pmb{\hat{X}_2})$ from \Cref{lem:rootTreeST}. Consider an arbitrary terminal $t$. Let $P_2$ for terminal $t$ be the shortest path from $t$ to $\hat{X}_1$ in $G[\hat{X}_2]$. Let $u_t$ be the sink of $P_2$ and let $P_1$ be the shortest path from $u_t$ to $r$ in $G[\hat{X}_1]$. Notice that $(\hat{x}_1, \pmb{\hat{x}_2})$ along with $r_2$ which sends one unit of flow from $t$ to $u_t$ along $P_2$ and $r_1$ which sends one unit of flow from $u_t$ to $r$ along $P_1$ for every $t$ is a feasible solution to \ref{LP:STEmaxLP}. Moreover, notice that cost of this solution is $\ct(\hat{x}_1, \pmb{x_2}) = \ct(X_1, \pmb{X_2}) \leq 2 \cdot \ct(O_1, \pmb{O_2})$ by \Cref{lem:rootTreeST}.
\end{proof}

Previous work of Gupta et al.\ \cite{GRS-FOCS04} shows that it is possible to round a fractional solution of \ref{LP:STEmaxLP} such that \emph{every} scenario has a good cost.
\begin{lemma}[\cite{GRS-FOCS04}]\label{lem:STRound}
A fractional solution $(x_1, \pmb{x_2})$  to \ref{LP:STEmaxLP} can be rounded in polynomial time to a feasible integral solution $(X_1, \pmb{X_2})$ s.t. $\cost(X_1, X_2^{(s)}) \leq 15 \cdot \cost(x_1, x_2^{(s)})$ for every $s$.
\end{lemma}

Since \Cref{cor:STCor} gives \ref{LP:STEmaxLP} has a good optimal solution, we can round \ref{LP:STEmaxLP} s.t. every scenario has a low cost. Now \Cref{lem:scenByScen} tells us that such a rounding preserves the cost of a solution for \trunc optimization. This gives the following theorem.
\begin{theorem}\label{thm:STTruncApx}
\emax Steiner tree can be $\left(\frac{30}{1-1/e}\right)$-approximated in polynomial time.
\end{theorem}
\begin{proof}
Our algorithm first solves \ref{LP:STEmaxLP} to get fractional solution  $(x_1, \pmb{x_2})$. Next, we apply \Cref{lem:STRound} to round $(x_1, \pmb{x_2})$ in polynomial time to give $(X_1, \pmb{X_2})$ as our solution. Thus, we have
\begin{align*}
\ct(X_1, \pmb{X_2}) &\leq 15 \cdot \ct(x_1, \pmb{x_2}) \bec{\Cref{lem:scenByScen}, \Cref{lem:STRound}}\\
&\leq 30 \cdot \ct(O_1, \pmb{O_2}), \bec{\Cref{cor:STCor}}
\end{align*}
\noindent where $(O_1, \pmb{O_2})$ is the optimal \trunc Steiner tree solution.
This implies we have a $30$-approximation algorithm for \trunc Steiner tree. Now by \Cref{thm:EmaxToCostB}, we have a $\left(\frac{30}{1-1/e}\right)$-approximation for \emax Steiner tree.

Lastly, each of our subroutines has a polynomial runtime by previous lemmas, and so we conclude that our algorithm has a polynomial runtime.
\end{proof}

%------------------------------------------------------------
\subsection{MST}\label{sec:MST}

In this section we give a randomized polynomial-time algorithm which with high probability has expected cost $O(\log n + \log m)$ times the optimal \emax minimum spanning tree (MST) on an $n$-node graph with $m$ different scenarios.

\begin{defn}[\emax MST]
We are given a graph $G = (V,E)$ where $|V| = n$, a set of $m$ scenarios $S_1, \ldots S_m$ where each scenario $S_s$ has an associated second-stage cost function $\cost_2^{(s)} : E \rightarrow \mathbb{Z}^+$  and a probability $p_s$. We are also given a first-stage cost function, $\cost_1 : E \rightarrow \mathbb{Z}^+$. We must provide a first stage solution $X_1 \subseteq E$ and a solution $X_2^{(s)} \subseteq E$ for every scenario $s$, which is feasible if $G[X_1 \cup X_2^{(s)}]$ spans $V$ for every $s$. The cost  for scenario $s$ in solution $(X_1, \pmb{X_2})$  is
\begin{align}
\cost(X_1, X_2^{(s)}) := \sum_{e \in X_1} \cost_1(e) + \sum_{e \in X_2^{(s)}} \cost_2^{(s)}(e).
\end{align}
The total cost for solution $(X_1, \pmb{X_2})$ is $\ce(X_1, \pmb{X_2}) := \E_{A \sim \pmb{p}} \left[ \max_{s \in A} \{\cost(X_1, X_2^{(s)})\} \right]$.
\end{defn}

 Our algorithm is based on the work of Dhamdhere et al.\ \cite{DRS-IPCO05} on two-stage stochastic MST. They give a rounding technique that produces integral solutions where every scenario has a cost close to the fractional cost. Using this rounding, and applying \Cref{lem:scenByScen}, we get an approximation algorithm for \trunc MST, which by \Cref{thm:EmaxToCostB} is also sufficient to approximate \emax MST.

 Notice that since \emax generalizes two-stage robust optimization, our \emax result gives a $O(\log n + \log m)$ approximation for two-stage robust MST as a corollary. To the best of our knowledge, this is the first non-trivial algorithm for two-stage robust MST.

Our algorithm is based on an LP. We have $m+1$ variables for each edge $e$, namely $x_{1}(e)$ and $x_{2}^{(s)}(e)$ for $s \in [m]$ indicating if we take $e$ in the first stage and in the second stage for scenario $s$, respectively. For a fractional solution $(x_1, \pmb{x_2})$, we define
\begin{align}
\cost(x_1, x_2^{(s)}) := \sum_e x_1(e) \cdot \cost_1(e) + x_2^{(s)}(e) \cdot \cost_2(e),
\end{align}
 which as described in Eq.\eqref{eq:fractionalCT}, defines $\ct(x_1, \pmb{x}_2)$ for fractional $(x_1, \pmb{x_2})$. Letting $\delta(S)$ be all edges with exactly one endpoint in $S \subseteq V$. The following is our LP.
\begin{align}
\label{LP:MSTLP}
\min  \qquad &\ct(x_1, \pmb{x_2}) \tag{MST LP}\\
\text{s.t.} \qquad & \sum_{e \in \delta(S)} \Big(x_{1}(e) + x_{2}^{(s)}(e) \Big) \geq 1 &\forall \emptyset \subset S \subset V, s \in [m] \notag \\
& x_1, \pmb{x_2} \geq 0	\notag
\end{align}
Note that an integral solution to \ref{LP:MSTLP} is a feasible solution for the \trunc MST problem as a set of edges with at least one edge leaving every cut is a spanning tree.\footnote{If such a solution has any cycles it is not necessarily an MST, though one can always delete an edge from such a cycle and improve the cost of the solution.} Also, although this LP has super-polynomial constraints, it is easy to obtain an efficient separation by solving min-cut; see Dhamdhere et al.\ \cite{DRS-IPCO05}.

We need the following  result of Dhamdhere et al.\ \cite{DRS-IPCO05}  to round \ref{LP:MSTLP} such that every scenario has a low cost.
\begin{lemma}[\cite{DRS-IPCO05}]\label{lem:MSTRounding}
It is possible to randomly round a feasible fractional solution $(x_1, \pmb{x_2})$ to \ref{LP:MSTLP} to an integral solution $(X_1, \pmb{X_2})$ in polynomial time s.t. with probability at least $1-\frac{1}{mn^2}$ for every scenario $s$ we have  $\E[\cost(X_1, X_2^{(s)})] \leq \cost(x_1, x_2^{(s)}) \cdot (40 \log n + 16 \log m)$. Here the expectation is taken over the randomness of our rounding and $m$ is the number of scenarios.
\end{lemma}

We can now design our approximation algorithm for \emax MST.
\begin{theorem}
There exists a randomized polynomial-time algorithm that with probability at least $1 - \frac{1}{mn^2}$ in expectation $O(\log n + \log m)$-approximates \emax MST where $n = |V|$ and $m$ is the number of scenarios.
\end{theorem}
\begin{proof}
Our algorithm starts by following \ref{LP:MSTLP} to get a fractional solution $(x_1, \pmb{x_2})$. Next, apply \Cref{lem:MSTRounding} to round $(x_1, \pmb{x_2})$ to an integral solution $(X_1, \pmb{X_2})$. Return $(X_1, \pmb{X_2})$.

Next consider the cost of $(X_1, \pmb{X_2})$. Let $(O_1, \pmb{O_2})$ be the optimal integral solution to our \trunc MST problem and let $(o_1, \pmb{o_2})$ be the corresponding characteristic vector. Notice that $(o_1, \pmb{o_2})$ is a feasible solution to \ref{LP:MSTLP}. Moreover, it is easy to verify that $\ct(o_1, \pmb{o_2}) = \ct(O_1, \pmb{O_2})$. Taking expectations over the randomness of our algorithm and applying \Cref{lem:scenByScen} and \Cref{lem:MSTRounding},  we have  with probability at least $1 - \frac{1}{mn^2}$ that
\begin{align*}
\E[\ct(X_1, \pmb{X_2})] &\leq (40 \log n + 16 \log m) \cdot \ct(o_1, \pmb{o_2}) \\
&= (40 \log n + 16 \log m) \cdot \ct(O_1, \pmb{O_2}).
\end{align*}
Thus, with probability at least $1 - \frac{1}{mn^2}$ our algorithm's expected \trunc cost is within $(40 \log n + 16 \log m)$ of the cost of the optimal \trunc MST solution. We conclude by \Cref{thm:EmaxToCostB} that with high probability in expectation our algorithm $O( \log n + \log m)$-approximates \emax MST.\footnote{Although \Cref{thm:EmaxToCostB} and \Cref{lem:scenByScen} do not explicitly account for an expectation taken over the randomness of an algorithm, it is easy to verify that the such an expectation does not affect these results.}

Our algorithm is trivially polynomial-time by the separability of our LP and \Cref{lem:MSTRounding}.
\end{proof}

%------------------------------------------------------------
\subsection{Min-Cut}\label{sec:minCut}
In this section we give a polynomial-time $\left(\frac{4}{1-1/e} \right)$-approximation for \emax min-cut.

\begin{defn}[\emax min-cut]
We are given a graph $G = (V,E)$, a root $r \in V$, a cost $c_e$ for edge $e$, and $m$ scenarios specified by terminals $t_1, \ldots, t_m \in V$. Each scenario $t_s$ has an associated probability $p_s$ and inflation factor $\sigma_s > 0$. We must provide a first stage solution $X_1 \subseteq E$ and a second-stage solution $X_2^{(s)} \subseteq E$ for each $s$. A feasible solution is one where $X_1 \cup X_2^{(s)}$ cuts $r$ from $t_s$ for every $s$. The cost for scenario $s$ in solution $(X_1, \pmb{X_2})$  is
\begin{align}
\cost(X_1, X_2^{(s)}) := \sum_{e \in X_1} c_e + \sigma_s \cdot \sum_{e \in X_2^{(s)}} c_e.
\end{align}
\noindent The total cost  of  solution $(X_1, \pmb{X_2})$ is $\ce(X_1, \pmb{X_2}) := \E_{A \sim \pmb{p}}\left [ \max_{s \in A} \{ \cost(X_1, X_2^{(s)}) \} \right ]$.
\end{defn}

% Throughout this section $\Pe$  stands for an instance of \emax min-cut, and $\Pt$  stands for its corresponding \trunc version.

We draw on past work on two-stage stochastic min-cut. In particular, we use the insight of Golovin et al.\ \cite{GGPRS-MP15} that, when approximating min-cut in a two-stage setting, it suffices to consider a relaxed version of the two-stage problem. In the second stage of this relaxed problem one does not pay the cost of completing their first stage solution. Rather, if the vertex corresponding to a scenario is not fully cut away in the first stage, in the second stage one must pay the full cost of cutting away that vertex in the original graph. The utility of this observation is that the relaxed problem can be captured by an LP (which  is not clear in general for  two-stage min-cut problems).

Thus, we first write an LP for the relaxed version of \trunc min-cut and then round it using ideas from Golovin et al.\ \cite{GGPRS-MP15}. We make use of \Cref{lem:scenByScen} in several places in our analysis to show that the optimal solution has certain structure, ultimately showing that such an algorithm $4$-approximates the \trunc version of min-cut. As shown in \Cref{thm:EmaxToCostB}, $4$-approximating \trunc min-cut is sufficient to $(\frac{8}{1-1/e})$-approximate \emax min-cut.

We let $\widehat{\Pt}$ be the previously mentioned relaxed problem; we use ``$\widehat{\hspace{10pt}}$'' to indicate objects and functions in this relaxed problem. Problem $\widehat{\Pt}$ is similar to $\Pt$ but the form of the second stage solution and cost of each scenario is different. In particular, we must give a first stage solution $X_1 \subseteq E$ and a second stage solution $\pmb{\hat{X_2}} \in \{0, 1\}^m$ indicating if each $t_s$ is cut away by our first stage solution. A solution $(X_1, \pmb{\hat{X_2}})$ is feasible if for every scenario $s$ we have  $\hat{X_2}^{(s)} = 1$ iff $X_1$ cuts $t_s$ from $r$. For scenario $s$ in this solution we pay
\begin{align}
\widehat{\cost}(X_1, \hat{X_2}^{(s)}) := \sum_{e \in X_1} \left[c_e \right] + \sigma_s (1 - \hat{X}_2^{(s)}) \cdot \text{cut-cost}(r, t_s),
\end{align}
\noindent where $\text{cut-cost}(r, t_s)$ is the minimum cost of a cut separating $r$ from $t_s$ in $G$.

%We analogously define the other relevant quantities for this relaxation. WLOG let our scenarios be indexed such that $\widehat{\cost}(X_1, \hat{X_2}^{(s)}) \geq \widehat{\cost}(X_1, \hat{X_2}^{(s+1)})$ and let $b$ be the smallest positive integer such that $\sum_{s=1}^b p_s \geq 1$. For integral $(X_1, \hat{X_2}^{(s)})$ we have
%\begin{align}\label{defn:fractionalCT}
%&\widehat{M}(X_1, \hat{X_2}^{(s)}) := [b]\\
%&\widehat{B}(X_1, \hat{X_2}^{(s)}) := \min_{s \in \widehat{M}(X_1, \hat{X_2}^{(s)})} \widehat{\cost}(X_1, \hat{X_2}^{(s)})\\
%&\widehat{\ct}(X_1, \hat{X_2}^{(s)}) := \widehat{B}(X_1, \hat{X_2}^{(s)}) + \sum_{s} p_s \cdot (\widehat{\cost}(X_1, \hat{X_2}^{(s)}) -  \widehat{B}(X_1, \hat{X_2}^{(s)}))^+
%\end{align}
%in $\widehat{\Pt}$. \enote{These definitions seem a bit redundant but I'm also worried that it would not be sufficiently clear without them}

  We capture $\widehat{\Pt}$ with an LP.  We have a variable $x_1(e)$ for each edge $e$  standing for whether we cut $e$ in the first stage and a variable $x_2^{(s)}$ for each $s$  standing for whether or not $t_s$ is cut from $r$ in the first stage. For fractional $(x_1, \pmb{x_2})$, we give its cost  in a given scenario $s$ as
\begin{align}
\widehat{\cost}(x_1, x_2^{(s)}) := \sum_e c_e \cdot x_1(e) + \sigma_s (1- x_2^{(s)}) \cdot \text{cut-cost}(r, t_s).
\end{align}
\noindent As described in Eq.\eqref{eq:fractionalCT}, this definition of $\widehat{\cost}(x_1, x_2^{(s)})$  also defines $\widehat{\ct}(x_1, \pmb{x_2})$ for fractional $(x_1, \pmb{x_2})$. Thus, we can now give our LP with $\widehat{\ct}(x_1, \pmb{x_2})$  as the objective value.
\begin{align}
\label{LP:MCLP}
\min  \qquad &\widehat{\ct}(x_1, \pmb{x_2}) \tag{MC LP}\\
\text{s.t.} \qquad & \sum_{e \in P} x_1(e) \geq x_2^{(s)} &\forall P \in \mathcal{P}_G(r, t_s), \forall s \label{line:cutConstr}\\
&0 \leq x_1(e), x_2^{(s)} \leq 1 &\forall s, e \in E,
\end{align}
\noindent where $\mathcal{P}_G(r, t_s)$ gives all paths from $r$ to $t_s$ in $G$. Although this LP has super-polynomial  constraints, it is easy to see that a polynomial-time algorithm for $s-t$ shortest path gives an efficient separation oracle. Hence, this LP  is solvable in polynomial time.

As Golovin at al.\ \cite{GGPRS-MP15} demonstrated, one can construct a feasible solution to \ref{LP:MCLP} which has cost in $\widehat{\Pt}$ roughly analogous to the optimal costs in $\Pt$ for every scenario.
\begin{lemma}[Lemma 3.1 in \cite{GGPRS-MP15}]\label{lem:MCFeasible}
Let $(O_1, O_2)$ be the optimal integral solution to $\Pt$. There exists a feasible integral solution to \ref{LP:MCLP}, $(x_1, \pmb{\hat{x}_2})$, such that $\widehat{\cost}(x_1, \hat{x}_2^{(s)}) \leq 2 \cdot \cost(O_1, O_2^{(s)})$ for every $s$. %$\cost(\hat{x}) \leq 2 \cdot \cost \left(X_1^* \right)$ and for every scenario $j$ such that $\hat{y}_j = 0$ we have $\text{cut-cost}_G(r, t_j) \leq 2 \cdot  \cost\left((X_2^*)^{(j)} \right)$.
\end{lemma}

Applying \Cref{lem:scenByScen} and the fact that the optimal solution to \ref{LP:MCLP} is certainly no more than $\widehat{\ct}(x_1, \pmb{\hat{x}_2})$ as given in \Cref{lem:MCFeasible}, we have the following corollary. This corollary shows that the above LP has an optimal value which is roughly the same as the cost of the optimal integral \trunc min-cut solution.
\begin{corollary}\label{lem:CCProbToProb}
Let $(o_1, \pmb{o_2})$ be the optimal fractional solution to \ref{LP:MCLP} and let $(O_1, \pmb{O_2})$ be the optimal solution to $\Pt$. We have $\widehat{\ct}(o_1, \pmb{o_2}) \leq 2 \ct(O_1, \pmb{O_2})$.
\end{corollary}

Having shown how \ref{LP:MCLP} has optimal cost analogous to the optimal solution to $\Pt$, we need only make use of the fractional solution to \ref{LP:MCLP} to construct an integral solution to $\Pt$. We do so with algorithm \textsc{MinCut\emax}. Roughly, this algorithm first cuts away all scenarios that were fractionally cut away by \ref{LP:MCLP} to an extent of at least $\frac{1}{2}$; its second stage solution is the minimum remaining cut for each scenario. See \Cref{alg:DSFromSchedule}.

\begin{algorithm}
	\caption{\textsc{MinCut\emax}}
	\label{alg:DSFromSchedule}
	\begin{algorithmic}
		\Statex \textbf{Input:} An instance of min-cut \emax
		\Statex \textbf{Output:} A solution to the input instance
		\State $(o_1, \pmb{o_2}) \gets \text{optimal fractional solution to }\ref{LP:MCLP}$
		\State $U = \left\{ t_s \mid o_2^{(s)} \geq \frac{1}{2}, s \in [m]  \right\}$
		\State $X_1 \gets \text{minimum $r-U$ cut in G}$
		\State $X_{2}^{(s)} \gets \text{minimum $r-t_s$ cut in $G \setminus X_1$}$ for each $s$
		\State \Return $\left(X_1, \pmb{X_2} \right)$
	\end{algorithmic}
\end{algorithm}

Given a set of first-stage edges, $X_1$, we let $\hat{X}_2^{(s)}(X_1)$ be the natural way to derive a second-stage solution for $\widehat{\Pt}$ from a first-stage solution. In particular,
\begin{align*}
\hat{X}_2^{(s)}(X_1) := \begin{cases}1 & \text{if $X_1$ cuts $t_s$ from $r$} \\ 0 &\text{o/w} \end{cases}
\end{align*}

 We now argue that \textsc{MinCut\emax} solves $\Pt$ at cost proportional to its completion to a solution to $\widehat{\Pt}$.
\begin{lemma}\label{lem:LPMCBound}
Let $(X_1, \pmb{X_2})$ be the returned values of \textsc{MinCut\emax}. We have
\begin{align*}
\ct(X_1, \pmb{X_2}) \leq \widehat{\ct}(X_1, \pmb{\hat{X}_2(X_1)}).
\end{align*}
\end{lemma}
\begin{proof}

We first argue that $\sum_{e \in X_2^{(s)}} c_e \leq (1 - \hat{X}_2^{(s)}(X_1)) \cdot \text{cut-cost}_G(r, t_s)$ for any $s$.  We case on whether $X_1$ cuts $t_s$ from $r$.
\begin{itemize}
\item If $X_1$ cuts $t_s$  from $r$ then we have  $\hat{X}_2^{(s)}(X_1) = 1$ and so trivially $\sum_{e \in X_2^{(s)}} c_e = 0$ meaning $\sum_{e \in X_2^{(s)}} c_e \leq (1 - \hat{X}_2^{(s)}(X_1)) \cdot \text{cut-cost}_G(r, t_s)$ .
\item If $X_1$ does not cut $t_s$ from $r$ then we have  $\hat{X}_2^{(s)}(X_1) = 0$. But $X_2^{(s)}$ is a minimum $r-t_s$ cut in $G \setminus X_1$ and $\text{cut-cost}_G(r, t_s)$ is the minimum cut cost in $G$ and so $\sum_{e \in X_2^{(s)}} c_e \leq \text{cut-cost}_G(r, t_j) = (1 - \hat{X}_2^{(s)}(X_1)) \cdot \text{cut-cost}_G(r, t_s)$.
\end{itemize}
Thus, for any $s$ we have
\begin{align}
\sum_{e \in X_2^{(s)}} c_e \leq (1 - X_2^{(s)}(X_1)) \cdot \text{cut-cost}_G(r, t_s). \label{eq:sahileq}
\end{align}

Next, notice that it follows that for any $s$ we have $\cost(X_1, X_2^{(s)}) \leq \widehat{\cost}(X_1, \hat{X}_2^{(s)}(X_1))$ since
\begin{align*}
\cost(X_1, X_2^{(s)}) &= \sum_{e \in X_1} c_e + \sigma_s \cdot \sum_{e \in X_2^{(s)}} c_e \\
&\leq \sum_{e \in X_1} c_e + \sigma_s (1 - X_2^{(s)}(X_1)) \cdot \text{cut-cost}_G(r, t_s) \bec{Eq.~\eqref{eq:sahileq}}\\
&= \widehat{\cost}(X_1, \hat{X}_2^{(s)}(X_1))
\end{align*}
Thus for any $s$ we have $\cost(X_1, X_2^{(s)}) \leq \widehat{\cost}(X_1, \hat{X}_2^{(s)}(X_1))$; applying \Cref{lem:scenByScen} gives \Cref{lem:LPMCBound}.
\end{proof}

Finally, combining previous lemmas we can prove that \textsc{MinCut\emax} efficiently approximates \emax min-cut.
\begin{theorem}
\textsc{MinCut\emax} is a polynomial-time $\left(\frac{4}{1-1/e} \right)$-approximation for \emax min-cut.
\end{theorem}
\begin{proof}

First, notice that by \Cref{lem:LPMCBound} we have
\begin{align}
\ct(X_1, \pmb{X_2}) \leq \widehat{\ct}(X_1, \pmb{\hat{X}_2(X_1)})\label{line:MC1}
\end{align}

Thus, for the remainder of this proof it will suffice to upper bound $\widehat{\ct}(X_1, \pmb{\hat{X}_2(X_1)})$. We begin by showing that $(X_1, \pmb{\hat{X}_2(X_1)})$ has cost in $\widehat{\Pt}$ roughly the same as the optimal solution to $\widehat{\Pt}$. In particular, we will show that $\widehat{\ct}(X_1, \pmb{\hat{X}_2(X_1)}) \leq 2 \widehat{\ct}(o_1, \pmb{o_2})$ where $(o_1, \pmb{o_2})$ is the optimal fractional solution to \ref{LP:MCLP}. Let $\bar{o}_1 = 2 o_1$.

First, we upper bound the cost of $X_1$ in $\widehat{\Pt}$ relative to $o_1$. We do so by first arguing that  $\bar{o}_1$ is a fractional $r-U$ cut. For $t_s \in U$ we have  $o_2^{(s)} \geq \frac{1}{2}$. It follows by \Cref{line:cutConstr} that for every path $P$ from $r$ to $t_s$ we have $\sum_{e \in P} o_1(e) \geq o_2^{(s)} \geq \frac{1}{2}$ and so we have for every path $P$ from $r$ to $t_s$ that $\sum_{e \in P} \bar{o}_1(e) = \sum_{e \in P} 2 o_1(e) \geq 2 o_2^{(s)} \geq 1$. Thus, $\bar{o}_1$ is a fractional $r-U$ cut. Next notice that since the minimum cut is a lower bound on any fractional cut and $\bar{o}_1$ is a fractional $r-U$ cut and $X_1$ is a minimum $r-U$ cut, we have
\begin{align}\label{eq:firstStagePHatT}
\sum_{e \in X_1}c_e \leq \sum_e \bar{o}_1(e) \cdot c_e = \sum_e 2 o_1(e) \cdot c_e
\end{align}

Next, we upper bound the cost of $\pmb{\hat{X}_2(X_1)}$ in $\widehat{\Pt}$ relative to $\pmb{o_2}$. First recall that $\hat{X}_2^{(s)}(X_1)$ is 1 if $X_1$ cuts $t_s$ from $r$ and 0 otherwise. %$t_j \in U$ and $1$ if $t_j \not \in U$.
If $\hat{X}_2^{(s)}(X_1) = 1$  we trivially have  $\text{cut-cost}_G(r, t_s) \cdot (1 - \hat{X}_2^{(s)}(X_1)) \leq \text{cut-cost}_G(r, t_s) \cdot (1 - o_2^{(s)})$. If $\hat{X}_2^{(s)}(X_1) = 0$ then we have  $X_1$ does not cut $t_s$ from $r$ which means that $t_s \not \in U$ and so $o_2^{(s)} < \frac{1}{2}$. It follows that in this case $\text{cut-cost}_G(r, t_s) \cdot (1 - \hat{X}_2^{(s)}(X_1)) = \text{cut-cost}_G(r, t_s)  \leq 2 \cdot \text{cut-cost}_G(r, t_s) \cdot (1 - o_2^{(s)})$. Thus, for any $s$ we have
\begin{align}\label{eq:secondStagePHatT}
\text{cut-cost}_G(r, t_s) \cdot (1 - \hat{X}_2^{(s)}(X_1)) \leq 2\cdot \text{cut-cost}_G(r, t_s) (1 - o_2^{(s)})
\end{align}

Since we have upper bound the first and second stage costs of $(X_1, \pmb{\hat{X}_2(X_1)})$ in $\widehat{\Pt}$ we can upper bound its total cost in $\widehat{\Pt}$. In particular, we have  for every $s$ it holds that
\begin{align*}
\widehat{\cost}(X_1, \hat{X}_2^{(s)}(X_1)) &= \sum_{e \in X_1} \left[c_e \right] + (1 - \hat{X}_2^{(s)}) \cdot \text{cut-cost}(r, t_s) \bec{dfn. of $\widehat{\cost}(X_1, \hat{X}_2^{(s)}(X_1))$}\\
 &\leq \sum_e 2 o_1(e) \cdot c_e  + 2\cdot \text{cut-cost}_G(r, t_s) (1 - o_2^{(s)}) \bec{Equations \ref{eq:firstStagePHatT}, \ref{eq:secondStagePHatT}}\\
  &= 2\widehat{\cost}(o_1, o_2^{(s)}) \bec{dfn.\ of $\widehat{\cost}(o_1, o_2^{(s)})$}
\end{align*}
Thus, for any scenario $s$ we have  $\widehat{\cost}(X_1, \hat{X}_2^{(s)}(X_1)) \leq 2\widehat{\cost}(o_1, o_2^{(s)})$ and so applying \Cref{lem:scenByScen} we get
\begin{align}
\widehat{\ct}(X_1, \pmb{\hat{X_2}(X_1)}) \leq 2 \widehat{\ct}(o_1, \pmb{o_2}). \label{line:MC2}
\end{align}

To complete our proof we notice that by \Cref{lem:CCProbToProb} we have
\begin{align}
2\widehat{\ct}(o_1, \pmb{o_2}) \leq 4 \ct(O_1, \pmb{O_2}). \label{line:MC3}
\end{align}

Combining \Cref{line:MC1}, \Cref{line:MC2} and \Cref{line:MC3} we conclude that
\begin{align*}
\ct(X_1, \pmb{X_2}) \leq 4 \ct(O_1, \pmb{O_2}).
\end{align*}

Thus, our algorithm, \textsc{MinCut\emax}, is a $4$-approximation for \trunc min-cut. Applying \Cref{thm:EmaxToCostB}, we conclude  \textsc{MinCut\emax} is a $\left(\frac{4}{1-1/e} \right)$-approximation for \emax min-cut.

A polynomial runtime follows from the fact that our algorithm needs to only solve LP \ref{LP:MCLP}, compute $U$ and then compute polynomially many min-cuts.
\end{proof}

% !TeX root = arXiv.tex
% !TEX root = arXiv.tex

%%%%%%%%%%%%%%%%%%%%%%%%%%%
\section{\emax k-Center}\label{sec:clusteringAppl}

In this section we give a constant approximation for \emax $k$-center.

\begin{defn}[\emax $k$-center]
We are given a metric space $(\met, \{c_{ij}\})$ over points \met, a set of scenarios $\{S_s\}_{s=1}^m$ where  $S_s$ corresponds to  client $s \in \met$, and a probability $p_s$ for each scenario. We must output $X \subseteq \met$ which is feasible if $|X| \leq k$. The cost we pay for solution $X$ is \[ \E_{A\sim \pmb{p}} \Big[\max_{s \in A} d(X, s) \Big],\] where $d(X,s) := \min_{i \in X} c_{is}$.
\end{defn}

Notice that unlike the preceding \emax problems, here we only provide a first stage solution. 
   \emax $k$-center can be phrased as a two-stage covering problem with a non-linear cost if we set the cost of any solution that opens a facility in the second stage to $\infty$.
 For this reason, we let $B(X)$ and $M(X)$ stand for $B(X, \emptyset)$ and $M(X, \emptyset)$ respectively as defined in Eq.\eqref{defn:MAndB}. 
 
Roughly, our algorithm works as follows. We draw on the intuition behind \Cref{thm:EmaxToCostB} that the expected max is well-approximated by truncating and summing values, and therefore truncate distances in the metric. We then solve a $k$-center-like LP on this truncated metric and use our LP solution  to cluster together nearby clients. Finally, in this clustered version of our problem we  run a $k$-median algorithm and return its solution as our solution for the  \emax $k$-center problem. Our techniques follow  Chakrabarty and Swamy \cite{CS-arXiv17} combined with  careful applications of \Cref{lem:eTrick} among other techniques to handle challenges unique to \emax k-center.

We begin by describing the LP. Define $f_B$ as the function that truncates at $B$, i.e.,
\begin{align*}
f_B(d) := \begin{cases} d \text{ ~if } d \geq B \\ 0 \text{~ otherwise.} \end{cases}
\end{align*}
\noindent Let $X^*$ be an optimal solution to \emax k-center, let $X_T^*$ be the optimal solution to the corresponding \trunc problem. Given $B$ as a guess of $B(X_T^*)$, our LP has a variable $x_i$   indicating if $i$ is a center and a variable $z_{is}$ indicating the extent to which we assign $s$ to $i$.
\begin{align*}
\label{LP:KCLLP}
\min\qquad &  \sum_s p_s \cdot \Big( \sum_i f_{B}(c_{is}) \cdot z_{is} \Big) \tag{$P_B$} \\
\text{s.t.} \qquad & \sum_i z_{is} \geq 1, \qquad \forall s \\
&0 \leq z_{is} \leq x(i),  \qquad \forall i, \forall s  \\
&\sum_{i}x_i \leq k	
\end{align*}
Let $\val(B)$ be the optimal value of \ref{LP:KCLLP} given parameter $B$. Let $\OPT$ be the cost of the optimal solution for the input \emax k-center problem.

Although we would like to use $B(X_T^*)$ as our value for $B$ in \ref{LP:KCLLP}, we do not  know  $B(X_T^*)$. For this reason, in the following lemma we argue how to efficiently compute a value that, up to constants,  works as well. 
The proof of this lemmas is deferred to \S\ref{sec:defProofsCluster}; roughly, the idea is to  take $B$ as the  best power of $(1 + \eps)$.

\begin{restatable}{lemma}{searchB}\label{lem:searchB}
There exists a polynomial-time algorithm which for a  given  $\eps>0$ and an input instance of \emax k-center, returns a $\hat{B}$ such that $\hat{B} \leq (1 + \eps) \left(\frac{\OPT}{1-1/e} \right)$ and $\val(\hat{B}) \leq  \left(\frac{3 \cdot \OPT}{1-1/e} \right)$. The algorithm's runtime is polynomial in $n$, $m$, and $\log_{1+\eps}\OPT$.
\end{restatable}

%\enote{Change $\hat{B}$ to $\OPT$; add details in proof}
%The crucial property of this LP that we make use of is that its optimal value is roughly the optimal value of its corresponding \emax k-center problem.
%\begin{lemma}
%The optimal value of~\ref{LP:KCLLP} is at most $\hat{B}$.
%\end{lemma}
%\begin{proof}
%Consider  the optimum solution  after the $1-1/e$ trick. Its cost looks like $p_1d_1 + p_2 d_2+ \ldots+ p_{\ell}d_{\ell}$, where $d_1 \geq d_2 \geq \ldots \geq d_{\ell}$ and $\sum_{i=1}^{\ell}p_i=1$. Since $\hat{B}$ is a guess for $\sum_{i=1}^{\ell}p_i d_i$, we know $d_{\ell} \leq \hat{B}$. Hence, if we substitute this optimum solution into~\ref{LP:KCLLP}, we pay at most $\hat{B}$.
%\end{proof}

Given a good value of $B$, we can describe our algorithm in full. Our algorithm first computes $\hat{B}$ as in \Cref{lem:searchB}. It next uses $P_{\hat{B}}$ to cluster clients. Let 
\[P_{\hat{B}}(s) :=  \sum_i f_{\hat{B}}(c_{is})  z_{is} \]
  be the cost of the scenario with client $s$ in $P_{\hat{B}}$. We sort scenarios in increasing order of $P_{\hat{B}}(s)$. For each client $s'$ initialize $\mathcal{P}_{s'}$ to $0$. Next, iterate through the clients. For client $s'$, if there exists a client $s$ s.t. $c_{ss'} \leq 2\hat{B}$  with $\mathcal{P}_s > 0$, then increment $\mathcal{P}_s$ by $p_{s'}$. Otherwise, set $\mathcal{P}_{s'}$ to $p_{s'}$. Let $\met' := \{s\in \met : \mathcal{P}_s > 0\}$ and let $\sigma: \met \rightarrow \met' $ be a function where $\sigma(s')$ denotes the client to whom we move client $s'$'s probability mass.

Now consider the weighted $k$-median instance consisting of clients $\met'$ with distances $\{c_{ij}\}$, where $s' \in \met'$ has weight $\mathcal{P}_{s'}$ and where one can choose centers only at points in $\met'$. Call this instance $\med$. Notice that weighted $k$-median can be reduced to unweighted $k$-median by just duplicating points and scaling costs by the appropriate amount. Run any $\alpha$-approximation for $k$-median on $\med$ and return the output as our solution to the input \emax k-center problem.

We now prove that our algorithm achieves a constant approximation in polynomial-time. Henceforth, let $(x,z)$ denote an optimal solution to $P_{\hat{B}}$ and as before let $\val(\hat{B})$ denote the value of this solution. Moreover, let $\val'(\hat{B}) := \sum_{s' \in \met'}  \mathcal{P}_{s'} \cdot P_{\hat{B}}(s')$ be the cost of $(x,z)$ applied to $\met'$.

We first show that clients in $\med$ are far apart which will allow us to argue that truncating distances at $\hat{B}$ does not affect distances.
\begin{lemma}\label{lem:farApart}
If $s_1',s_2' \in \met'$ then $c_{s_1's_2'} > 2 \hat{B}$.
\end{lemma}
\begin{proof}
WLOG suppose that $s_1'$ is considered before $s_2'$ in the clustering above. Moreover, suppose for the sake of contradiction that $c_{s_1's_2'} \leq 2 \hat{B}$. Notice that when we examine $s_2'$ we will assign $s_2'$ to $s_1'$; i.e., $\sigma(s_2') = s_1'$. However, it follows that $\mathcal{P}_{s_1'} = 0$ and as such $s_1' \not \in \met'$ by definition of $\met'$, a contradiction.
\end{proof}

We next show that the value of our LP solution to $P_{\hat{B}}$ only decreases in cost when applied to $\met'$.
\begin{lemma}\label{lem:optToOptP}
$\val'(\hat{B}) \leq \val(\hat{B})$.
\end{lemma}
\begin{proof}
When we cluster points in increasing order of $P_{\hat{B}}(s)$, we have that the cluster to which any given client is reassigned is always of lesser cost in $P_{\hat{B}}$; i.e., if $s \in \sigma^{-1}(s')$ then 
\begin{align} \label{eq:sortedLP}
 \sum_{i} f_{\hat{B}}(c_{is'}) \cdot z_{is'} \leq  \sum_{i} f_{\hat{B}}(c_{is}) \cdot z_{is}.\end{align}
 Thus, we can get $ \val(\hat{B})$ equals
\begin{align*}
 \sum_{s \in \met} p_s \cdot  \sum_{i} f_{\hat{B}}(c_{is}) \cdot z_{is} &=  \sum_{s' \in \met'} \Big[ \sum_{s \in \sigma^{-1}(s')} p_s  \cdot \sum_{i}   f_{\hat{B}}(c_{is} )\cdot z_{is} \Big]\\
&\geq \sum_{s' \in \met'} \Big[ \sum_{s \in \sigma^{-1}(s')} p_{s}  \sum_{i} f_{\hat{B}}(c_{is'}) \cdot z_{is'}\Big] \bec{Eq.~\eqref{eq:sortedLP}}\\
&= \sum_{s' \in \met'} \Big[\Big( \sum_{i} f_{\hat{B}}(c_{is'}) \cdot z_{is'} \Big) \cdot \sum_{s \in \sigma^{-1}(s')} p_s \Big]\\
&= \sum_{s' \in \met'} \Big[\Big( \sum_{i} f_{\hat{B}}(c_{is'}) \cdot z_{is'} \Big) \cdot \mathcal{P}_{s'} \Big] 
\\
&=  \val'(\hat{B}).	\qedhere
\end{align*}
\end{proof}

Next we show that there exists a solution to $\med$ of cost about $\val(\hat{B})$.
\begin{lemma}\label{lem:fracISol}
The optimal solution to $\med$ has cost at most $O(\val(\hat{B}))$.
\end{lemma}
\begin{proof}
Again, let $(x, z)$ be the optimal solution to $P_{\hat{B}}$. We prove this lemma by constructing a fractional solution $(x', z')$ to the LP of $\med$ of value at most $2\cdot \val(\hat{B})$. The LP of $\med$ is as follows and has variables analogous to \ref{LP:KCLLP} (we overload $(x', z')$ here to stand for the variables in our LP along with the feasible solution for our LP that we will construct).
\begin{align}
\label{LP:KMLP}
\min\qquad &  \sum_{s_1' \in \met'} \mathcal{P}_s \Big( \sum_{s_2' \in \met'} c_{s_2's_1'} \cdot z'_{s_2's_1'} \Big) \nonumber \tag{k-M LP}\\
\text{s.t.} \qquad & \sum_{s_2' \in \met'} z'_{s_2's_1'} \geq 1, \qquad \forall s_1'\\
&0 \leq z'_{s_2's_1'} \leq x'(s_2'),  \qquad \forall s_2', \forall s_1'\\
&\sum_{s' \in \met '}x'(s') \leq k
\end{align}

To construct $(x', z')$, we do the following. We first define a new clustering such that every point in $\met$ goes to the closest point in $\met'$. Formally, for $s'\in \met'$, we define $F_{s'} := \{ i\in \met :  s' = \argmin_{s' \in \met'}{c_{is'}} \}$. Intuitively, our solution $(x', z')$  reroutes services that were provided by facilities in $F_{s'}$ to $s's$. 

Let $x'(s'):= \sum_{i \in F_{s'}} x(i)$. That is, we open a facility at $s'$ by summing up the facilities in $P_{\hat{B}}$ that were clustered to $s'$. We let $z_{s_1's_2'}' = \sum_{i \in F_{s_1'}} z_{is_2'}$. That is, we assign $s_2'$ to $s_1'$ to the extent that $P_{\hat{B}}$ assigned $s_2'$ to points clustered with $s_1'$.

We now prove the feasibility of $(x', z')$ for \ref{LP:KMLP}. Since for every client $s$ in $P_{\hat{B}}$ we know $\sum_{i} z_{is} \geq 1$, we have that every client in $\met'$ is serviced: For every client $s_1' \in \met'$ it holds that $\sum_{s_2' \in \met'}z_{s_2's_1'}' = \sum_{s_2' \in \met'} \sum_{i \in F_{s_2'}} z_{is_1'} = \sum_{i} z_{is_1'} \geq 1$. Moreover, we open no more than $k$ centers fractionally since in $P_{\hat{B}}$ we have that
$\sum_i x(i) \leq k$ and
\begin{align*}
\sum_{s' \in \met'} x'(s') \quad = \quad \sum_{s' \in \met'} \sum_{i \in F_{s'}} x(i) \quad = \quad \sum_{i} x(i) \quad \leq \quad k.
\end{align*}

\noindent Also, no client is serviced by an unopened facility: in $P_{\hat{B}}$ we have that $0 \leq z_{is} \leq x(i)$. Hence,  for any $s_1', s_2' \in \met'$ we have $z'_{s_1's_2'}= \sum_{i \in F_{s_1'}} z_{is_2'} \leq \sum_{i \in F_{s_1'}} x(i) = x'(s_1')$.

Lastly, we bound the objective value of $(x', y')$. By \Cref{lem:farApart} we know that for $s_1' \neq s_2'$, where $s_1', s_2' \in \met'$, it holds that $c_{s_1's_2'} > 2\hat{B}$. We first show that  if $i \in F_{s_2'}$ then $c_{is_1'} > \hat{B}$. This is because if $c_{is_1'} \leq \hat{B}$ then $c_{is_2'}\leq \hat{B}$ since $s_2'$ is the closest point in $\met'$ to $i$ by $i\in F_{s_2'}$. Now by triangle inequality, we get $c_{s_1's_2'} \leq c_{s_1'i} + c_{is_2'} \leq 2\hat{B}$, which is a contradiction to \Cref{lem:farApart}. Hence,  $c_{is_1'} > \hat{B}$  implies $f_{\hat{B}}(c_{is_1'}) = c_{is_1'}$, and 
\begin{align}
c_{s_1's_2'} &\leq 2  c_{is_1'} = 2 \cdot f_{\hat{B}}(c_{is_1'}) \label{eq:cjjP},
\end{align}
where $i \in F_{s_2'}$ for $s_2' \neq s_1'$ and $s_1', s_2' \in \met'$.

Thus, the value of $(x',z')$ in \ref{LP:KMLP} is
\begin{align*}
\sum_{s_1' \in \met'} \mathcal{P}_{s_1'} \left( \sum_{s_2' \in \met'} c_{s_2's_1'} \cdot z'_{s_2's_1'} \right) & = \sum_{s_1', s_2' \in \met'} \mathcal{P}_{s_1'} \cdot c_{s_1's_2'} \cdot z'_{s_2's_1'} \\
&\leq \sum_{s_1', s_2' \in \met' : s_1' \neq s_2'} \mathcal{P}_{s_1'} \cdot c_{s_1's_2'} \cdot z'_{s_2's_1'} \bec{ $c_{ss} = 0$}\\
&= \sum_{s_1', s_2' \in \met' : s_1' \neq s_2'} \sum_{i \in  F_{s_2'}} \mathcal{P}_{s_1'} \cdot c_{s_1's_2'} \cdot z_{is_1'} \bec{definition of $z_{s_1's_2'}'$}\\
&\leq 2 \cdot \sum_{s_1', s_2' \in \met' : s_1' \neq s_2'} \sum_{i \in  F_{s_2'}} \mathcal{P}_{s_1'} \cdot f_{\hat{B}}(c_{is_1'}) \cdot z_{is_1'} \bec{Eq.~\eqref{eq:cjjP}} \\
&\leq 2 \cdot \sum_{s_1' \in \met'} \sum_{i} \mathcal{P}_{s_1'} \cdot f_{\hat{B}}(c_{is_1'}) \cdot z_{is_1'} \\
&= 2 \cdot \val'(\hat{B})\\
& \leq 2 \cdot \val(\hat{B}) \bec{\Cref{lem:optToOptP}}.
\end{align*}

Thus, since $(x', z')$ is feasible and has cost at most $2 \val(\hat{B})$ we know that the optimal solution to \ref{LP:KMLP} has cost at most $2 \val(\hat{B})$. Moreover, since past work has demonstrated that \ref{LP:KMLP} has a constant integrality gap---e.g.\cite{CGTS0-JCSS02} shows it is at most $20/3$---we conclude that the optimal solution to $\med$ has cost at most $O(\val(\hat{B}))$.
\end{proof}

Next, we show that any solution to $\med$ is a good solution to our \emax k-center problem.
\begin{lemma}\label{lem:intISol}
An integer solution to $\med$ of cost $C$ solves the input \emax $k$-center problem with cost at most $C + 4 \cdot \hat{B}$.
\end{lemma}
\begin{proof} Observe that the primary difference between  \emax $k$-center problem and $\med$ is that the former is on $\met$ while the latter is on $\met'$. Roughly, this lemma is true because any point in $\met$ is at most $2\hat{B}$ from a point in $\met'$. 

Let $X$ be our $\med$ solution of cost $C$.
We have  that the cost of $X$ as a \emax solution is
\begin{align*}
\E_{A} [\max_{s \in A} \{ d(s, F)\} ]&\leq \sum_{s \in M(F)} p_s \cdot d(s, F) \bec{\Cref{lem:eTrick}}\\
&= \sum_{s' \in \met'} \sum_{s \in \sigma^{-1}(s') \cap M(F)} p_s \cdot d(s, F)\\
&\leq \sum_{s' \in \met'} \sum_{s \in \sigma^{-1}(s') \cap M(F)} p_s \cdot \left(c_{ss'} + d(s', F) \right) \bec{triangle inequality}\\
&\leq \sum_{s' \in \met'} \sum_{s \in \sigma^{-1}(s') \cap M(F)} p_s \cdot \left(2 \hat{B} + d(s', F) \right) \bec{$c_{ss'} \leq 2 \hat{B}$ if $\sigma(s) = s'$}\\
&\leq 4 \hat{B} + \sum_{s' \in \met'}  \left(d(s', F) \right)  \sum_{s \in \sigma^{-1}(s') \cap M(F)} p_s \bec{$\sum_{s \in M(F)} p_s < 2$}\\
&\leq 4 \hat{B} + \sum_{s' \in \met'}  \left(d(s', F) \right)  \mathcal{P}_{s'} \bec{definition of  $\mathcal{P}_{s'}$}\\
&\leq 4 \hat{B} + C  \bec{definition of $C$}. \quad\qedhere
\end{align*}
\end{proof}

Lastly, we conclude the approximation factor of our algorithm.

\begin{theorem}
\emax k-center can be $O(1)$-approximated in polynomial time.
\end{theorem}
\begin{proof}
By \Cref{lem:fracISol} the optimal solution to $\med$ has cost at most $O(\val(\hat{B}))$. Now, applying our $\alpha$-approximation algorithm for $k$-median to $\med$  results in an integer solution of cost at most $\alpha \cdot O(\val(\hat{B}))$. By \Cref{lem:intISol} such an integer solution solves \emax $k$-center with cost at most $\alpha\cdot  O(\val(\hat{B})) + 4 \hat{B}$. Applying \Cref{lem:searchB} we then have that our solution costs at most 
\[ \alpha\cdot  O\Big(\Big(\frac{3}{1-1/e} \Big) \cdot  \OPT \Big) + 4 (1 + \eps) \Big(\frac{\OPT}{1-1/e} \Big).\] 
Letting $\eps$ be any constant  $>0$ and using an $\alpha=O(1)$ approximation algorithm for k-median---e.g., \cite{CGTS0-JCSS02}---we conclude that our solution has cost at most $O(OPT)$.
 
Lastly, we argue that our algorithm runs in polynomial time. Solving our LPs and performing clustering are trivially poly-time. Running our $\alpha$-approximation for $k$-median is poly-time by assumption, so we conclude the polynomial runtime of our algorithm. 
\end{proof}

%END IGNORE

% !TeX root = arXiv.tex
% !TEX root = arXiv.tex

\newcommand{\cha}{\ch^1}
\newcommand{\chb}{\ch^2}
\newcommand{\cea}{\ce^1}
\newcommand{\ceb}{\ce^2}

%%%%%%%%%%%%%%%%%%%%%%%%%%%
\section{Reducing Stochastic, Demand-Robust, and \hybrid  to \emax} \label{sec:RedToEmax}
In this section we show how to use an $\alpha$-approximation algorithm in the \emax model to design an $\alpha$-approximation in the stochastic, the demand-robust, and the \hybrid models (Theorem \ref{thm:RedToEmax}). We believe our \hybrid model gives a clean way of modeling inaccuracy of the input distribution in a stochastic two-stage setting.

\paragraph{Our New \hybrid Model.} In our \hybrid two-stage covering model---as in our stochastic model---we are  given a distribution $\calD$ over scenarios from which exactly one scenario realizes. However, we are also given a \emph{caution parameter} $\rho$ which specifies the inaccuracy of $\calD$. A low value of $\rho$ signals that the second-stage realization is expected to be close to the stochastic model while a large $\rho$ signals that any scenario may occur. Specifically, in the second-stage w.p. $(1-\rho)$ the scenario we must pay for realizes from the input distribution, and w.p. $\rho$ it is chosen adversarially. Thus, the cost for solution $(X_1, \pmb{X_2})$ is a convex combination of the stochastic and robust objectives using $\rho$:
\begin{align} \label{eq:HybridModel}
\ch(X_1, \pmb{X_2}) :=  \rho \cdot \cro(X_1, \pmb{X_2}) + (1-\rho) \cdot \cst(X_1, \pmb{X_2}).
\end{align}
Note that for $\rho=0$ we recover the  stochastic model and for $\rho=1$ we recover the  demand-robust model. Thus, for $\rho \in (0,1)$, the \hybrid model smoothly interpolates between these two models and so it suffices to prove Theorem \ref{thm:RedToEmax} for only the \hybrid model. Recall that the theorem that we will prove is as follows.

\RedToEmax*

%The main idea of our reduction from a \hybrid problem, $\Ph$, to an \emax problem, $\Pe$, is as follows. For each scenario in our original \hybrid problem we create two scenarios in our \emax problem; one to represent the demand-robust cost of this scenario and one to represent its stochastic cost. The scenario to represent the robust cost has probability one but its cost is dampened by $\rho$. The scenario to represent the stochastic cost has a sufficiently low probability so that the independent Bernoulli trials of these scenarios are effectively disjoint. Moreover, the costs of the scenarios to represent the stochastic costs are inflated enough to make up for the dampened probability and to ensure that if any one of the scenarios occurs it is the most expensive.

%We first note that it suffices to prove this reduction only for the \hybrid model. This is because  \hybrid  captures  stochastic and demand-robust  for $\rho=0$ and $\rho=1$, respectively. We make two ``scaled" copies of every scenario to capture the stochastic  and   demand-robust parts of the \hybrid objective. The scenario to represent the robust cost has probability one but its cost is dampened by $\rho$. The scenario to represent the stochastic cost has a sufficiently low probability so that the independent Bernoulli trials of these scenarios are effectively disjoint. Moreover, the costs of the scenarios to represent the stochastic costs are inflated enough to make up for the dampened probability.

\noindent \textbf{Theorem \ref{thm:RedToEmax} Intuition.}
The main idea of our reduction from a \hybrid problem, $\Ph$, to a \emax problem, $\Pe$, is as follows. For each scenario in our original \hybrid problem we create two scenarios in our \emax problem; one to represent the demand-robust cost of this scenario and one to represent its stochastic cost. The scenario to represent the robust cost has probability one but its cost is dampened by $\rho$. The scenario to represent the stochastic cost has a sufficiently low probability so that the independent Bernoulli trials of these scenarios are effectively disjoint. Moreover, the costs of the scenarios to represent the stochastic costs are inflated  to make up for the dampened probability, and to ensure that they are more expensive than the demand-robust scenarios.

\subsection*{Our Reduction}  Let $\Ph$ be a \hybrid problem with distribution $\calD$ over scenarios $\{\overline{S}_1,\ldots, \overline{S}_m\}$, a caution parameter $\rho$, and a {linear} cost function %$\overline{\cost}$. 
\[ \ch(X_1,X_2) = \cha(X_1) + \chb(X_2),
\]
where $\cha(X_1) :=\ch(X_1, \emptyset)  $ and $\chb(X_2) :=\ch(\emptyset, X_2)$.
We produce a \emax instance $\Pe$ with $2m$ scenarios $S_1,S_2, \ldots, S_{2m} $ where the first $m$ scenarios encode the demand-robust cost and the last $m$ scenarios encode the stochastic  cost. For $s\in [m]$, the covering constraints for $S_s$ and $S_{m+s}$ are the same as $\overline{S}_s$.  
We set the cost function  
\[\ce(X_1,X_2^{(s)})  := \cea(X_1) + \ceb(X_2^{(s)})\]
such that the first stage cost is the same as the \hybrid first stage, i.e., $	\cea(X_1) = \cha(X_1)$.
In the second stage for  $s\in [m]$, we set
\begin{align*}
 \quad \ceb(X_2^{(s)}) &:= \rho \cdot \chb(X_2^{(s)}) \\
  \text{and} \quad \ceb(X_2^{(m+s)}) &:= \gamma(1-\rho) \cdot \chb(X_2^{(s)}),
\end{align*}
where $\gamma\geq 1$ is a sufficiently large scaling factor for our stochastic scenarios. In particular, $\gamma$ satisfies that for all $s,s' \in [m]$,
\begin{align} \label{eq:defnGamma}
\gamma (1-\rho) \cdot \chb( X_2^{(s)}) > \rho \cdot \chb(  X_2^{(s')}).
\end{align}
For $s\in [m]$, we set the probabilities in $\Pe$ to be $p_s :=1$ and $p_{m+s}=\frac{ \calD(s)}{\gamma}$. 

Notice that the stochastic copy of a scenario always costs more than the demand-robust one, i.e., $\ce(X_1, X_2^{(s')}) > \ce(X_1, X_2^{(s)})$ for $s \leq m$ and $s' > m$. See \Cref{fig:hybridRed} for an illustration. 

\begin{figure} 
\centering
\begin{tikzpicture}[scale=.62]

%ORIGINAL SCENARIOS
\draw [thick] (-4,-1) rectangle (4,7);  \node [below right] at (-4, 7) {$\Ph$};
\drawScen[$\overline{S}_1$](-2,2)[4,.4]
\drawScen[$\overline{S}_2$](0,2)[8,.4]
\drawScen[$\overline{S}_3$](2,2)[2,.2]

\draw[->, ultra thick] (4,2) -- (7, 2); \node [above] at (5.5, 2) {Reduction};

\draw [thick] (7,-1) rectangle (22,7); \node [below right] at (7, 7) {$\Pe$};

%DR SCENARIOS
\draw [rounded corners=5pt, dashed] (8,0) rectangle  (14,4); \node [below] at (11, 0) {Demand-robust scenarios};
\drawScen[$S_1$](9,2)[1,1]
\drawScen[$S_2$](11,2)[2,1]
\drawScen[$S_3$](13,2)[.5,1]

%STOCHASTIC SCENARIOS
\draw [rounded corners=5pt, dashed] (15,0) rectangle (21,4); \node [below] at (18, 0) {Stochastic scenarios};
\drawScen[$S_4$](16,2)[6,.2]
\drawScen[$S_5$](18,2)[12,.2]
\drawScen[$S_6$](20,2)[3,.1]

\end{tikzpicture}
\caption{Reduction from $\Ph$ for $\rho = .25$ to $\Pe$ for $m=3$. Circles: scenarios. Costs are drawn for fixed solution $(X_1, \pmb{\overline{X}_2})$ for $\Ph$ and its corresponding solution for $\Pe$, $(X_1, \pmb{X_2})$  where for $s \equiv s' \mod m$ we define $X_2^{(s)} := \overline{X}_2^{(s')}$. Red rectangles in $\Ph$: $\ch(X_1, \overline{X}_2^{(s)})$. Red rectangles in $\Pe$: $\ce(X_1, X_2^{(s)})$. Blue rectangles in $\Ph$: $\mathcal{D}(s)$ of $\overline{S}_s$. Blue rectangles in $\Pe$: $p_s$. $\gamma = 2$ and $\cha(X_1) = \cea(X_1)= 0$. Demand-robust scenarios have high probability and stochastic scenarios have high cost.}
\label{fig:hybridRed}
\end{figure}

%--------------------------------------------------------------------------------
\subsection*{Our Proof of \Cref{thm:RedToEmax}} Observe that by linearity of  costs  in \emax,  our first stage and second stage costs can be separated. This implies for any solution $(X_1, \pmb{X_2})$ to $\Pe$, we have
%For a $(X_1, \pmb{X_2})$ be any solution to $\Pe$. We have
\begin{align} \label{lem:sepFirstStageCost}
\ce(X_1, \pmb{X_2}) = \cea(X_1)  + \E_{A} \Big[\max_{s \in A}\Big\{\ceb(X_2^{(s)})\Big\} \Big].
\end{align}

\IGNORE{
\begin{proof}
By linearity of \cost and the definition of $\overline{\cost}$, we have
\begin{align} \label{eq:linearityOfCost}
\ce(X_1, \pmb{X_2}) &= \E_{A} \Big[\max_{s \in A} \Big\{ \cost(X_1,X_2^{(s)}) \Big\} \Big]\\
   &=   \cost(X_1,\emptyset) + \E_{A} \Big[\max_{s \in A}\Big\{\cost(\emptyset,X_2^{(s)})\Big\} \Big] \notag \\
 & =   \overline{\cost}(X_1,\emptyset) + \E_{A} \Big[\max_{s \in A}\Big\{\cost(\emptyset,X_2^{(s)})\Big\} \Big].
\end{align}
\end{proof}
}

To prove \Cref{thm:RedToEmax}, in \Cref{lem:hybLB}  we show that for our reduction the optimal $\Pe$ solution costs less than the optimal  \hybrid solution. We defer the proof of this lemma to \S\ref{sec:defProofsHyb}. In \Cref{lem:hybUB} we show the converse direction (up to a small factor).

\begin{restatable}{lemma}{hybLB}\label{lem:hybLB}
Let $(H_1, \pmb{\overline{H}_2})$ be the optimal solution to $\Ph$ and let $(H_1, \pmb{H_2})$ be its natural interpretation in $\Pe$, i.e., for $s \equiv s' \mod m$ we define $H_2^{(s)} := \overline{H}_2^{(s')}$. We have $\ce(H_1, \pmb{H_2}) \leq \ch(H_1, \pmb{\overline{H}_2})$.
\end{restatable}
\noindent  The main idea of the proof is to use a union bound over the stochastic scenarios.

Next, we show that any solution costs more in $\Pe$ than in $\Ph$ (up to small multiplicative factors).
\begin{lemma} \label{lem:hybUB}
Let $(X_1, \pmb{X_2})$ be any solution to $\Pe$ and let $(X_1, \pmb{\overline{X}_2})$ be its natural interpretation in $\Ph$, i.e., $\overline{X}_2^{(s)} := \argmin  \Big\{\chb(X_2^{(s)}), \chb(X_2^{(s+m)}) \Big\}$. We have $\ch(X_1, \pmb{\overline{X}_2}) \leq \left( 1 - \frac{m}{\gamma} \right)^{-1} \cdot \ce(X_1, \pmb{X_2})$.
\end{lemma}
\begin{proof}
Notice that by the definition of $\overline{X}_2$, we know 
\begin{align}
\ch(X_1, \overline{X}_2^{(s)}) = \min \Big\{ \ch(X_1, X_2^{(s)}), \ch(X_1, X_2^{(s+m)}) \Big\}. \label{eq:minNatInt}
\end{align}

We first lower bound  the second stage cost of \emax using the inclusion-exclusion principle on the number of stochastic scenarios that realize into $A$. Recollect, $A$ always contains the first $m$ demand-robust scenarios because their probability is $1$. Since a stochastic scenario (if it realizes) always costs more than all the demand-robust scenarios  by Eq.~\eqref{eq:defnGamma}, we get
%Now to prove our upper bound on $\ch(X_1, \pmb{\overline{X}_2})$ we notice that
\begin{align*}
 \E_{A} \Big[\max_{s \in A} \Big\{ \ceb(X_2^{(s)}) \Big\} \Big]  &\geq      \Pr_{A}[A = [m]] \cdot \max_{s \in [m]} \Big\{ \ceb(X_2^{(s)}) \Big\} + \sum_{s> m } \Pr_{A}[s \in A] \cdot  \ceb(X_2^{(s)}) \notag \\
& \qquad  -  \sum_{s,s'> m } \Pr_{A}[s \in A] \Pr_{A }[s' \in A]  \cdot  \max \Big\{\ceb(X_2^{(s)}),\ceb(X_2^{(s')})\Big\}.
\end{align*}
Now using $ \Pr_{A}[A = [m]] = \prod_s \Big(1-\frac{\calD(s)}{\gamma} \Big) \geq  (1-\frac{1}{\gamma})^m$ and  the definition of $\ceb$, we get
\begin{align}
 \E_{A} \Big[\max_{s \in A} \Big\{ \ceb(X_2^{(s)}) \Big\} \Big]   &\geq      \Big(1-\frac{1}{\gamma} \Big)^m \cdot \rho \max_{s \in [m]} \Big\{ \chb( X_2^{(s)})  \Big\} + \gamma (1-\rho) \sum_{s>m} p_s \cdot \chb( X_2^{(s)}) \notag \\
& \qquad  -  \gamma (1-\rho) \cdot \sum_{s,s'> m }p_s p_{s'} \cdot  \max \Big\{ \chb(X_2^{(s)}),\chb(X_2^{(s')}) \Big\}. \label{eq:LowerEmax}
\end{align}
Since for any $a,b\geq 0$, we have $\max\{a,b\} \leq a+b$, we can bound
\begin{align*}
	\sum_{s,s'> m }p_s p_{s'}   \max \Big\{ \chb(X_2^{(s)}),\chb(X_2^{(s')}) \Big\} &\leq 2  \sum_{s,s'> m }p_s p_{s'}    \chb(X_2^{(s)}) = \frac{2}{\gamma} \cdot \sum_{s> m }p_s    \chb(X_2^{(s)}),
\end{align*}
where the last equality uses $\sum_{s' > m} p_{s'} =  \sum_{s' > m} \frac{\calD({s'})}{\gamma} = \frac{1}{\gamma} $.
Combining this with Eq.~\eqref{eq:minNatInt} and Eq.~\eqref{eq:LowerEmax}, 
\begin{align*}
& \E_{A} \Big[\max_{s \in A} \Big\{ \ceb(X_2^{(s)}) \Big\} \Big] \\
 &\geq  \Big(1-\frac{1}{\gamma} \Big)^m \cdot \rho \max_{s \in [m]} \Big\{ \chb( X_2^{(s)})  \Big\} + \gamma (1-\rho) \Big( \sum_{s>m} p_s \cdot \chb(X_2^{(s)})  
  -   \frac{2}{\gamma} \sum_{s>m}p_s  \cdot   \chb(X_2^{(s)}) \Big)\\
    & \geq \Big(1-\frac{m}{\gamma} \Big) \cdot \rho \max_{s \in [m]} \Big\{ \chb( \overline{X}_2^{(s)})  \Big\} +  (1-\rho) \Big(1-\frac{m}{\gamma} \Big) \cdot \E_{s\sim \calD} \Big[ \chb(\overline{X}_2^{(s)})  \Big] ,
\end{align*}
where the last inequality uses $(1-\frac{1}{\gamma})^m \geq (1-\frac{m}{\gamma})$ along with $(1-\frac{2}{\gamma}) \geq (1-\frac{m}{\gamma} )$ and $\gamma  p_s = {\calD(s)}$.
Our upper bound follows when we combine the last inequality with Eq.~\eqref{lem:sepFirstStageCost}  to get
\begin{align*}
	 \ce(X_1, \pmb{X_2}) &\geq  \cea (X_1) +  \Big(1-\frac{m}{\gamma} \Big)  \Big( \rho \max_{s \in [m]} \Big\{ \chb ( \overline{X}_2^{(s)})  \Big\} + (1-\rho)   \E_{s\sim \calD} \Big[ \chb ( \overline{X}_2^{(s)})  \Big]  \Big)\\
& \geq \Big(1-\frac{m}{\gamma} \Big) \Big(  \rho \max_{i \in [m]} \Big\{ \ch (X_1,  \overline{X}_2^{(s)})  \Big\} + (1-\rho)  \cdot \E_{s\sim \calD} \Big[ \ch (X_1, \overline{X}_2^{(s)})  \Big]  \Big)\\
&= \Big(1-\frac{m}{\gamma} \Big) \cdot \ch(X_1, \pmb{\overline{X_2}}). \qquad \qedhere
\end{align*}
\end{proof}

We now conclude the proof of the main theorem by combining the $\alpha$-approximation algorithm with \Cref{lem:hybLB} and \Cref{lem:hybUB} which allow us to move between \emax and \hybrid via our reduction.
\begin{proof}[Proof of \Cref{thm:RedToEmax}]
Since  \hybrid  captures the stochastic and  demand-robust models  it suffices to prove the theorem only for \hybrid. 

%\enote{Change hat notation so not overline over hat}
Consider an input \hybrid problem $\Ph$ with optimal solution $(H_1, \pmb{\overline{H}_2})$. Suppose we have an $\alpha$-approximation  algorithm for  \emax. To design an $\alpha$-approximation for $\Ph$, we simply run our reduction to get an instance of \emax problem $\Pe$, run our \emax approximation algorithm on $\Pe$ to get back $({A_1}, \pmb{{A_2}})$, and then return $({A_1}, \pmb{\overline{{A_2}}})$ where $\overline{{A_2}}$ is the natural interpretation of ${A_2}$ as a solution for $\Ph$. In particular, 
\begin{align*}
\overline{{A}}_2^{(s)} := \argmin \Big\{ \chb({A}_2^{(s)}), \chb({A}_2^{(s+m)}) \Big\}.
\end{align*}
Now consider the cost of our returned solution for $\Ph$. Let $(H_1, \pmb{H_2})$ be its natural interpretation in $\Pe$, i.e., for $s \equiv s' \mod m$ we define $H_2^{(s)} := \overline{H}_2^{(s')}$. 
\ifFULL We have
\begin{align*}
\ch({A}_1, \pmb{\overline{{A}}_2}) & \leq \Big( 1 - \frac{m}{\gamma} \Big)^{-1} \ce({A}_1, \pmb{{A}_2}) \bec{\Cref{lem:hybUB}}\\
&  \leq \alpha \Big( 1 - \frac{m}{\gamma} \Big)^{-1} \ce(E_1, \pmb{E_2}) \bec{$({E}_1, \pmb{{E}_2})$ an $\alpha$-approx}\\
&  \leq \alpha \Big( 1 - \frac{m}{\gamma} \Big)^{-1} \ce(H_1, \pmb{H_2}) \bec{$(E_1, \pmb{E_2})$ minimizes $\ce$}\\
&  \leq \alpha \Big( 1 - \frac{m}{\gamma} \Big)^{-1} \ch(H_1, \pmb{\overline{H}_2}) \bec{\Cref{lem:hybLB}}.
\end{align*}
\else
By \Cref{lem:hybUB} we know 
\[ \ch({A}_1, \pmb{\overline{{A}}_2}) \leq \Big( 1 - \frac{m}{\gamma} \Big)^{-1} \ce({A}_1, \pmb{{A}_2}). \] 

Let $(E_1, \pmb{E_2})$  be the optimal solution of \emax. 
Since $({A}_1, \pmb{{A}_2})$ is an $\alpha$-approximation, 
\[ \ce({A}_1, \pmb{{A}_2}) \leq  \alpha \cdot  \ce(E_1, \pmb{E_2}).\]
 Combining the above two equations, we get 
 \[	\ch({A}_1, \pmb{\overline{{A}}_2}) \leq \alpha \cdot \Big( 1 - \frac{m}{\gamma} \Big)^{-1}   \ce(E_1, \pmb{E_2}) ,
  \]
  Since $(E_1, \pmb{E_2})$ is the optimal solution of \emax,  $ \ce(E_1, \pmb{E_2}) \leq \ce(H_1, \pmb{H_2})$. By Lemma~\ref{lem:hybLB},
  \[ \ce(H_1, \pmb{H_2}) \leq \ch(H_1, \pmb{H_2}) .
  \]
\fi
Thus, $({A}_1, \pmb{\overline{{A}}_2})$ is an $\alpha \cdot \left( 1 - \frac{m}{\gamma} \right)^{-1}$-approximation for $\Ph$. Choosing $\gamma \rightarrow  \infty$ and noting that the reduction is polynomial-time, the theorem follows.
\end{proof}

\section*{Appendix}
\appendix
% !TeX root = arXiv.tex
% !TEX root = arXiv.tex

\section{Deferred Proofs of \S\ref{sec:trunc}}\label{sec:defProofsTrunc}

%---1-1/e TRICK PROOF---
\eTrick*
\begin{proof}
We begin by showing the lower bound on $\E_{A \sim \pmb{Y}} \left[ \max_{s \in A}v_s \right]$. Let $M := [b]$.
Consider the new set of probabilities
\begin{align}
p'_s &= \begin{cases} 1 - \sum_{s < b} p_s &\text{if } s = b \\p_s &\text{otherwise} \end{cases}
\end{align}
and let $\pmb{Y}'$ be the corresponding Bernoulli r.v.s. Notice that $\sum_{s \in M} p_s' = 1$.

Since $p_s' \leq p_s$, clearly we have that $\E_{A \sim \pmb{Y}} \left[ \max_{s \in A}v_s \right] \geq \E_{A \sim \pmb{Y}'} \left[ \max_{s \in A}v_s \right]$. Thus, we will focus on lower bounding $\E_{A \sim \pmb{Y}'} \left[ \max_{s \in A}v_s \right]$. The probability that no element of $M$ is in $A$ when drawn from $\pmb{Y}'$ is 
\begin{align*}
\prod_{s \in M} (1 - p_s') \quad \leq \quad e^{-\sum_{s \in M} p_s'} \quad  = \quad \frac{1}{e} 
\end{align*}
because $1-x \leq e^{-x}$ and $\sum_{s \in M} p_s' = 1$.
It follows that 
\begin{align*}
\E_{A \sim \pmb{Y}} \Big[ \max_{s \in A}v_s \Big] &\geq \E_{A \sim \pmb{Y}'} \left[ \max_{s \in A}v_s \right]\\
 &\geq \Big(1 - \frac{1}{e} \Big) \E_{A \sim \pmb{Y}'} \Big[ \max_{s \in A}v_s \mid \text{at least 1 element from $M$ in $A$} \Big]\\
& \geq \Big(1 - \frac{1}{e} \Big) \E_{A \sim \pmb{Y}'} \Big[\max_{s \in A}v_s \mid \text{exactly 1 element from $M$ in $A$} \Big]\\
& = \Big(1 - \frac{1}{e} \Big)\sum_{s \in M} v_s \frac{p_s'}{\sum_{i \in M}p_s'} \quad = \Big(\frac{1 - 1/e}{1} \Big) \sum_{s \in M}p_s' v_s \bec{$\sum_{s \in M}p_s' = 1$}.
\end{align*}

Thus, we have that
\begin{align*}
\E_{A \sim \pmb{Y}} \Big[ \max_{s \in A}v_s \Big] & \geq  \Big(1 - \frac{1}{e} \Big) \sum_{s \in M}p_s' v_s\\
& = \left(1 - \frac{1}{e} \right) \sum_{s \in M}p_s' \Big((v_s - v_b)^+ + v_b \Big) \bec{$v_s \geq v_b$ for $s \in M$}\\
& \geq \Big(1 - \frac{1}{e} \Big) \Big(v_b  + \sum_{s \in M}p_s' \Big((v_s - v_b)^+ \Big) \Big) \bec{$1 = \sum_{s \in M} p_s'$}\\
& \geq \Big(1 - \frac{1}{e} \Big) \Big(v_b  + \sum_{s \in M}p_s \Big((v_s - v_b)^+ \Big) \Big) \bec{$(v_b - v_b)^+ = 0$}\\
& = \Big(1 - \frac{1}{e} \Big) \Big(v_b  + \sum_{s}p_s \Big((v_s - v_b)^+ \Big) \Big) \bec{$v_s > v_b$ iff $s \in M$}
\end{align*}
which gives our lower bound. 

We now show the upper bound. Recall $x^+ := \max(x, 0)$. Notice that  we have for any $t$,
\begin{align}
\label{eq:truncatanateUBound}
\max(x, y) \leq t + (x - t)^+ + (y - t)^+ .
\end{align}
In particular, Eq.~\eqref{eq:truncatanateUBound} follows because the RHS in each of the following cases is always $\geq \max\{x,y\}$.
\begin{itemize}
\item if $t \geq \max\{x,y\}$ we get $t$ for the RHS.
\item if $t \geq x$ and $t < y$  we get $t + y - t = y = \max\{x,y\}$ for the RHS; the symmetric case also holds.
\item if $t < x$ and $t < y$ we get $t + x - t + y - t = x + y - t \geq \max\{x,y\}$ for the RHS.
\end{itemize}
It is easy to verify that this holds for a $\max$ of more than two inputs; i.e.\ for a set $S$ of reals we have $\max(S) \leq t + \sum_{s \in S} (s - t)^+$. Thus, we have
\begin{align*}
\E_{A \sim \pmb{Y}} \Big[ \max_{s \in A}v_s \Big]  & \leq \E_{A \sim \pmb{Y}} \Big[ v_b + \sum_{s \in A} (v_s - v_b)^+\Big] ~~ = ~~  v_b + \E_{A \sim \pmb{Y}} \Big[\sum_{s \in A} (v_s - v_b)^+\Big]  \bec{Eq.~\eqref{eq:truncatanateUBound}}\\
%& =  v_b + \E_{A \sim \pmb{Y}} \Big[\sum_{s \in A} (v_s - v_b)^+\Big] \\
& = v_b + \E_{A \sim \pmb{Y}} \Big[\sum_{s \in A \cap M} (v_s - v_b)^+  +  \sum_{s \in A \cap (X \setminus M)} (v_s - v_b)^+ \Big] \\
& = v_b + \E_{A \sim \pmb{Y}} \Big[\sum_{s \in A \cap M} (v_s - v_b)^+  \Big] \bec{$v_s > v_b$ iff $s \leq b$}\\
& = v_b + \E_{A \sim \pmb{Y}} \Big[\sum_{s \in A \cap M} (v_s - v_b)  \Big] \bec{$v_s \geq v_b$ for $s \in M$}\\
& = v_b +\sum_{s \in M} p_s \cdot (v_s - v_b)\\
&= v_b + \sum_{s} p_s \cdot (v_s - v_b)^+ \bec{$v_s > v_b$ iff $s \in M$},
%& = v_b  + \sum_{s \in M} p_s v_s - v_b \cdot \sum_{s \in M} p_s\\
%& \leq v_b - v_b \cdot 1 + \sum_{s \in M} p_s v_s \bec{$\sum_{s \in M} p_s \geq 1$}\\
%&= \sum_{s \in M} p_s v_s\\
\end{align*}
which is exactly the desired upper bound.
\end{proof}

\ifFULL\else
\altTruncForm*
\convexityLemmaProof
\fi

\xEToTrunc*
\begin{proof} We have
\begin{align*}
\ce(X_1, \pmb{X_2}) &= \E_{A} [\max_{s \in A} \{\cost(X_1, X_2^{(s)}) \}] \\
& \leq B(X_1, \pmb{X_2}) + \sum_{s } p_s \cdot \left( \cost(X_1, X_2^{(s)}) - B(X_1, \pmb{X_2}) \right)^+ \bec{\Cref{lem:eTrick}}\\
&= \ct(X_1, \pmb{X}_2) \bec{\Cref{lem:altTruncForm}}.
\end{align*}
\end{proof}

\TOptToEOpt*
\begin{proof}
We have 
\begin{align*}
&\quad \ct(T_1, \pmb{T_2}) \\
&\leq \ct(E_1, \pmb{E_2}) \bec{$(T_1, \pmb{T_2})$ minimizes $\ct$}\\
&= \min_{B} \Big[B + \sum_{s} p_s \cdot (\cost(E_1, E_2^{(s)}) - B)^+\Big]\\
&\leq B(E_1, \pmb{E_2}) + \sum_{s} p_s \cdot (\cost(E_1, E_2^{(s)}) - B(E_1, \pmb{E_2}))^+ \\
&\leq \left( \frac{1}{1 - 1/e} \right ) \E_{A} [\max_{s \in A} \{\cost(E_1, E_2^{(s)})\} ] \bec{\Cref{lem:eTrick}}\\
&= \left( \frac{1}{1 - 1/e} \right ) \ce(E_1, \pmb{E_2}).
\end{align*}
\end{proof}

\section{Deferred Proofs of \S\ref{sec:clusteringAppl}}\label{sec:defProofsCluster}
We use the following lemma in the proof of \Cref{lem:searchB}.
\begin{lemma}\label{lem:searchBStarObs}
If $B \geq \OPT$ then $\val\left(\frac{B}{1-1/e}\right) \leq \left( \frac{3}{1-1/e} \right) \cdot \OPT \leq \left( \frac{3}{1-1/e} \right) \cdot B$.
\end{lemma}
\begin{proof}
Again, let $X_T^*$ be the optimal solution to the \trunc problem corresponding to our \emax k-center problem. First, notice that by \Cref{lem:TOptToEOpt} we have $ B(X_T^*) + \sum_s p_s \cdot (d(s, X_T^*) - B(X_T^*))^+ = \cost_T(X_T^*) \leq \left(\frac{1}{1-1/e} \right) \ce(X^*) = \left(\frac{1}{1-1/e} \right) \OPT$ and so we have $B(X_T^*) \leq  \OPT \cdot \left(\frac{1}{1-1/e} \right) $.

Now suppose $B \geq \OPT$. Let $z_{is}^*$ be $1$ if $X_T^*$ assigns $s$ to $i$ and $0$ otherwise and let $x^*(i)$ be $1$ if $i \in X_T^*$ and $0$ otherwise. Since $(x^*, z^*)$ is feasible for $P_{\left(\frac{B}{1-1/e} \right)}$ we have
\begin{align}
\val\Big(\frac{B}{1-1/e} \Big) &\leq \sum_s p_s \Big( \sum_i f_{(\frac{B}{1-1/e} )}(c_{is})\cdot z_{is}^* \Big) \nonumber\\
&\leq  \sum_s p_s \Big( \sum_i f_{B(X_T^*) }(c_{is} )\cdot z_{is}^* \Big) \label{line:BtoBStar}\\
&= \sum_{s \in M(X_T^*)} p_s \cdot (d(s, X_T^*)) \label{line:aLine} \\
&\leq \sum_{s \in M(X_T^*)} p_s \cdot ((d(s, X_T^*) - B(X_T^*))^+ +B(X_T^*)) \label{line:useTrunc} \\
&< 2B(X_T^*) + \ct(X_T^*) \label{line:sumOfProbs}\\
& \leq 2B(X_T^*) + \Big(\frac{1}{1-1/e} \Big) \cdot \OPT \label{line:useTrick},
\end{align}
where Eq.\eqref{line:BtoBStar} follows since $\left( \frac{1}{1-1/e} \right) \cdot B \geq \left( \frac{1}{1-1/e} \right) \cdot \OPT \geq  B(X_T^*)$, Eq.\eqref{line:aLine} follows by definition of $M(X_T^*)$, Eq.\eqref{line:useTrunc} follows by $x \leq (x-t)^+ + t$, Eq.\eqref{line:sumOfProbs} follows since $\sum_{s \in M(X_T^*)} p_s < 2$ and Eq.\eqref{line:useTrick} follows by \Cref{lem:TOptToEOpt}. Lastly, since $B(X_T^*) \leq  \OPT \cdot \left(\frac{1}{1-1/e} \right) $ by assumption, we have $\val\left(\frac{B}{1-1/e} \right) \leq  \left(\frac{3}{1-1/e} \right) \cdot \OPT$.
\end{proof}

\searchB*
\begin{proof}
Our algorithm to return $\hat{B}$ is as follows: let $\bar{B}$ be $(1 + \eps)^i$ where $i$ is the smallest $i \in \mathbb{Z}^+$ such that $\val\left(\frac{(1 + \eps)^i}{1-1/e}\right) \leq \left(\frac{3}{1-1/e} \right) \cdot (1 + \eps)^i$. Thus, we have
%\begin{enumerate}
\begin{enumerate}[itemsep=0ex,topsep=0pt,parsep=0pt]
\item $\bar{B} \leq (1 + \eps) \cdot \OPT$ by \Cref{lem:searchBStarObs} and how we choose $\bar{B}$;
\item $\val \left(\frac{\bar{B}}{1-1/e} \right) \leq \frac{3}{1-1/e} \cdot \bar{B}$  trivially by how we choose $\bar{B}$.
\end{enumerate}
Lastly, we return $\hat{B} := \frac{ \bar{B}}{1 - 1/e}$. By (1) and the definition of $\hat{B}$ we have that $\hat{B} \leq (1 + \eps) \left(\frac{\OPT}{1-1/e} \right)$. Moreover, if $\bar{B} < \OPT$ then by (2) we have $\val(\hat{B}) \leq \left(\frac{3}{1-1/e}\right) \OPT$ and if $\bar{B} \geq \OPT$ then by \Cref{lem:searchBStarObs} we have $\val (\hat{B}) \leq \left(\frac{3}{1-1/e}\right) \OPT$.
\end{proof}

\section{Deferred Proofs of \S\ref{sec:RedToEmax}}\label{sec:defProofsHyb}
\hybLB*

\begin{proof}
	We first bound the second stage cost of \emax  depending on whether a stochastic scenario ($s>m$) realizes into $A$ or not. Recollect, $A$ always contains the first $m$ demand-robust scenarios because their probability is $1$.  Since a stochastic scenario (if it realizes) always costs more than all the demand-robust scenarios  by Eq.~\eqref{eq:defnGamma}, we get
	\begin{align*}
	\E_{A} \Big[\max_{s \in A}\Big\{ \ceb(H_2^{(s)}) \Big\} \Big]  &  \leq    \Pr_{A }[A = [m]] \cdot \max_{s \in [m]} \Big\{ \ceb(H_2^{(s)})\Big\} + \sum_{s> m } \Pr_{A}[s \in A] \cdot  \ceb( H_2^{(s)}) \\
	% & \leq  1 \cdot \max_{s \in [m]} \Big\{\chb(H_2^{(s)}) \Big\} + \sum_{s> m} p_s \cdot  \chb(H_2^{(s)}) \\
	& \leq  1 \cdot \max_{s \in [m]} \Big\{\ceb(\overline{H}_2^{(s)}) \Big\} + \sum_{s \in [m]} p_s \cdot  \ceb(\overline{H}_2^{(s)}) \\
	& =    \max_{s \in [m]} \Big\{ \rho \cdot \chb( \overline{H}_2^{(s)})\Big\} + \frac{1}{\gamma}\E_{s\sim \calD} \Big[ \gamma (1-\rho) \cdot \chb(\overline{H}_2^{(s)}) \Big],
	\end{align*}
	where the last equality uses the definition of $\ceb$ and that $p_{s}=\frac{\calD(s)}{\gamma}$.
	Now using Eq.~\eqref{lem:sepFirstStageCost},  
	\begin{align*}
	\ce(H_1, \pmb{H_2})   & \leq   \cha(H_1) + \rho \cdot \max_{s \in [m]} \Big\{ \chb( \overline{H}_2^{(s)})\Big\} + (1-\rho) \cdot \E_{s\sim \calD} \Big[ \chb(  \overline{H}_2^{(s)}) \Big] \\
	&= \ch(H_1, \pmb{\overline{H}_2}). 	\qedhere
	\end{align*}	
\end{proof}

\bibliographystyle{alpha}
\bibliography{conferencestrings,bib}

\end{document}